\documentclass[11pt,a4paper,reqno,intlimits,sumlimits]{amsart}

\usepackage{amssymb}
\usepackage{color}
\usepackage[mathscr]{euscript}
\usepackage{xspace}
\usepackage{enumerate}
\usepackage{epsfig,rotating}
\input{xy}\xyoption{all}

\DeclareMathAlphabet{\mathpzc}{OT1}{pzc}{m}{it}

\begin{document}

\theoremstyle{plain}
\newtheorem{theorem}{Theorem}[section]
\newtheorem{lemma}[theorem]{Lemma}
\newtheorem{proposition}[theorem]{Proposition}
\newtheorem{claim}[theorem]{Claim}
\newtheorem{corollary}[theorem]{Corollary}
\newtheorem{axiom}{Axiom}

\theoremstyle{definition}
\newtheorem{remark}[theorem]{Remark}
\newtheorem{note}{Note}[section]
\newtheorem{definition}[theorem]{Definition}
\newtheorem{example}[theorem]{Example}
\newtheorem*{ackn}{Acknowledgements}
\newtheorem{assumption}{Assumption}
\newtheorem{approach}{Approach}
\newtheorem{critique}{Critique}
\newtheorem{question}{Question}
\newtheorem{aim}{Aim}
\newtheorem*{asa}{Assumption ($\mathbf{A}$)}
\newtheorem*{asp}{Assumption ($\mathbb{P}$)}
\newtheorem*{ass}{Assumption ($\mathbb{S}$)}
\renewcommand{\theequation}{\thesection.\arabic{equation}}
\numberwithin{equation}{section}

\renewcommand{\thefigure}{\thesection.\arabic{figure}}
\numberwithin{equation}{section}

\newcommand{\Law}{\ensuremath{\mathop{\mathrm{Law}}}}
\newcommand{\loc}{{\mathrm{loc}}}

\let\SETMINUS\setminus
\renewcommand{\setminus}{\backslash}

\def\stackrelboth#1#2#3{\mathrel{\mathop{#2}\limits^{#1}_{#3}}}

\newcommand{\prozess}[1][L]{{\ensuremath{#1=(#1_t)_{0\le t\le T}}}\xspace}
\newcommand{\prazess}[1][L]{{\ensuremath{#1=(#1_t)_{t\ge0}}}\xspace}
\newcommand{\pt}[1][N]{\ensuremath{\P_{#1}}\xspace}
\newcommand{\tk}[1][N]{\ensuremath{T_{#1}}\xspace}
\newcommand{\dd}[1][]{\ensuremath{\ud{#1}}\xspace}

\newcommand{\scal}[2]{\ensuremath{\langle #1, #2 \rangle}}
\newcommand{\set}[1]{\ensuremath{\left\{#1\right\}}}
\def\lev{L\'{e}vy\xspace}
\def\lk{L\'{e}vy--Khintchine\xspace}
\def\lib{LIBOR\xspace}
\def\mg{martingale\xspace}
\def\smmg{semimartingale\xspace}
\def\alm{affine LIBOR model\xspace}
\def\alms{affine LIBOR models\xspace}
\def\dalms{defaultable \alms}

\def\half{\frac{1}{2}}

\def\F{\ensuremath{\mathcal{F}}}
\def\bD{\mathbf{D}}
\def\bF{\mathbf{F}}
\def\bG{\mathbf{G}}
\def\bH{\mathbf{H}}
\def\R{\ensuremath{\mathbb{R}}}
\def\Rp{\mathbb{R}_{\geqslant0}}
\def\Rm{\mathbb{R}_{\leqslant 0}}
\def\C{\ensuremath{\mathbb{C}}}
\def\U{\ensuremath{\mathcal{U}}}
\def\I{\mathcal{I}}

\def\P{\ensuremath{\mathrm{I\kern-.2em P}}}
\def\Q{\mathbb{Q}}
\def\E{\ensuremath{\mathrm{I\kern-.2em E}}}

\def\ott{{0\leq t\leq T}}
\def\idd{{1\le i\le d}}

\def\icc{\mathpzc{i}}
\def\ecc{\mathbf{e}_\mathpzc{i}}

\def\uk{u_{k+1}}
\def\vk{v_{k+1}}

\def\e{\mathrm{e}}
\def\ud{\ensuremath{\mathrm{d}}}
\def\dt{\ud t}
\def\ds{\ud s}
\def\dx{\ud x}
\def\dy{\ud y}
\def\dv{\ud v}
\def\dsdx{\ensuremath{(\ud s, \ud x)}}
\def\dtdx{\ensuremath{(\ud t, \ud x)}}

\def\lsnc{\ensuremath{\mathrm{LSNC-}\chi^2}}
\def\nc{\ensuremath{\mathrm{NC-}\chi^2}}
\newcommand{\cD}{{\mathcal{D}}}
\newcommand{\cF}{{\mathcal{F}}}
\newcommand{\cG}{{\mathcal{G}}}
\newcommand{\cH}{{\mathcal{H}}}
\newcommand{\cK}{{\mathcal{K}}}
\newcommand{\cM}{{\mathcal{M}}}
\newcommand{\cT}{{\mathcal{T}}}
\newcommand{\ha}{{\mathbb{H}}}
\newcommand{\indik}{{\mathbf{1}}}
\newcommand{\ifdefault}[1]{\ensuremath{\mathbf{1}_{\{\tau \leq #1\}}}}
\newcommand{\ifnodefault}[1]{\ensuremath{\mathbf{1}_{\{\tau > #1\}}}}

\title{A tractable LIBOR model with default risk}

\author{Zorana Grbac}
\author{Antonis Papapantoleon}

\address{D\'epartment Math\'ematiques, Universit\'e d'\'Evry Val d'Essonne, 23
         boulevard de France, F-91037 \'Evry Cedex, France}
\email{zorana.grbac@univ-evry.fr}

\address{Institute of Mathematics, TU Berlin, Stra\ss e des 17. Juni 136,
         10623 Berlin, Germany}
\email{papapan@math.tu-berlin.de}

\thanks{We thank Monique Jeanblanc for several interesting discussions and
comments, and two anonymous referees for their detailed reports. Z.~G.
acknowledges the financial support from an AMAMEF short visit Grant (No. 2856).
Her research benefited from the support of the `Chaire Risque de cr\'edit' and
F\'ed\'eration Bancaire Fran\c caise. A.~P. acknowledges the financial support
from the DFG Research Center \textsc{Matheon}}

\date{}\maketitle\pagestyle{myheadings}\frenchspacing

\begin{abstract}
We develop a model for the dynamic evolution of default-free and defaultable
interest rates in a LIBOR framework. Utilizing the class of affine processes,
this model produces positive LIBOR rates and spreads, while the dynamics are
analytically tractable under defaultable forward measures. This leads to
explicit formulas for CDS spreads, while semi-analytical formulas are derived
for other credit derivatives. Finally, we give an application to counterparty
risk.
\end{abstract}

\section{Introduction}

The current financial crisis has brought default risk to the forefront of
attention, with large corporations and even countries being on the verge of
bankruptcy. This has led to a renewed demand for credit derivatives, which can
be used for hedging (or even for speculative) purposes; see the December 2011 issue of
the BIS Quarterly Review.

In this work, we present a tractable model for default-free and defaultable
interest rates and study the pricing of credit derivatives in this framework.
More precisely, we work in a discrete tenor framework and use the \textit{LIBOR
rate} as the risk-free rate. Of course, LIBOR is not considered risk-free any
longer, see e.g. \cite{Crepey_Grbac_Nguyen_2011} or
\cite{Filipovic_Trolle_2011}, but one can simply replace the \lib with the
``true'' risk-free rate in today's markets. Next, we consider a corporation that
issues bonds which are subject to default risk. The riskiness of these bonds is
reflected in their pre-default values, and we use them to derive the
\textit{defaultable LIBOR rate}, following \cite{Schoenbucher00} and
\cite{EberleinKlugeSchoenbucher06}. The defaultable LIBOR rate can be
interpreted as the effective rate a corporation pays for borrowing money, which
typically equals the LIBOR rate plus a (stochastic) spread (see also
\eqref{libor-spread}).

Therefore, the aim of this paper is to develop an analytically tractable model
for the joint evolution of default-free and defaultable LIBOR rates. The
classical models for LIBOR rates, based on the seminal articles by
\cite{SandmannSondermannMiltersen95}, \cite{BraceGatarekMusiela97}
and \cite{Jamshidian97}, are known to suffer from severe intractability
problems. This has led to a multitude of approximation methods; see
\cite[Ch.~10]{GatarekBachertMaksymiuk06} for an overview. These numerical
problems are propagated in the defaultable framework; let us just mention that
closed-form expressions do not exist even in the simple Brownian framework.
\cite{EberleinKlugeSchoenbucher06} use the so-called frozen drift
approximation to derive prices for derivatives, but this is well-known to
perform poorly for long maturities and high volatilities; we refer to
\cite{PapapantoleonSchoenmakersSkovmand10} for numerical experiments and
alternative approximation schemes.

In order to overcome these problems, we work in the framework of the
\textit{affine LIBOR models} recently introduced by
\cite{KellerResselPapapantoleonTeichmann09}. These models are based on the wide
and flexible class of affine processes, see
\cite{DuffieFilipovicSchachermayer03}, and have very appealing properties:
LIBOR rates are positive, the model is arbitrage-free and the dynamics remain
tractable under forward measures. The last property allows for semi-analytical
pricing of many interest rate derivatives using Fourier transforms. We extend
these models to the defaultable setting in a fashion that: (i) default-free and
defaultable rates are positive, (ii) the market is free of arbitrage and (iii)
the dynamics of rates are analytically tractable under restricted defaultable
forward measures. In this framework, we can derive fully explicit formulas for
CDS rates, while other credit derivatives admit semi-analytical pricing
formulas.

The defaultable affine LIBOR model we develop belongs to the reduced form
approach for modeling credit risk, where the default time is modeled as the
first jump of a Cox process, with a given cumulative hazard process. This
approach was studied by \cite{ElliottJeanblancYor00} and
\cite{JeanblancLeCam08a} from a theoretical perspective. For concrete examples
of cumulative hazard process models we refer to
\cite{EberleinKlugeSchoenbucher06} and \cite{KokholmNicolato10}, and
further references therein.

The paper is organized as follows. In Section \ref{alm} we recall the definition
and some important properties of affine processes and provide an overview of the
affine LIBOR models. In Section \ref{dalm-sec} we present the extension to the
defaultable setting. In particular, we introduce defaultable LIBOR rates and
discuss the no-arbitrage conditions in this set-up. This discussion follows
along the lines of \cite{EberleinKlugeSchoenbucher06}. Then we construct a
model for the dynamics of default intensities using a suitable family of
exponentially-affine processes. The conditions for fitting the initial term
structure of defaultable bonds are provided and analyzed. As an interesting
consequence, we obtain that the hazard process in this model is an affine
transformation of the driving affine process at tenor dates, which provides a
direct link to hazard process models for credit risk. In Section
\ref{dalm-price} we derive valuation formulas for various credit derivatives. An
explicit pricing formula is presented for credit default swaps and
semi-analytical formulas are derived for options on defaultable bonds. Finally,
we propose an application of the defaultable affine LIBOR model to counterparty
risk and study the pricing of vulnerable options. The pricing formulas are again
semi-analytical, based on Fourier transforms.

\section{The default-free affine LIBOR model}
\label{alm}

We provide below a brief overview of the construction and the main results of
the affine LIBOR model; for more details and proofs we refer to
\cite{KellerResselPapapantoleonTeichmann09}.

\subsection{Affine processes}

Let $(\Omega,\F,\bF,\P)$ denote a complete stochastic basis, where
$\bF=(\F_t)_{t\in[0,T]}$, and let $0 < T \le \infty$ denote some, possibly
infinite, time horizon. We consider a process $X$ of the following type:
\begin{asa}\label{assumption-affine}
Let \prozess[X] be a conservative, time-homogene\-ous, stochastically
continuous Markov process taking values in $D=\Rp^d$, and $(\P_x)_{x\in D}$ a
family of probability measures on $(\Omega,\F)$, such that $X_0 = x$
$\P_x$-almost surely, for every $x \in D$. Setting
\begin{align}
\mathcal{I}_T
:= \set{u\in\R^d: \E_x\big[\e^{\scal{u}{X_T}}\big] < \infty,
        \,\,\text{for all}\; x \in D},
\end{align}
we assume that
\begin{itemize}
\item[(i)] $0 \in \I_T^\circ$, where $\I_T^\circ$ denotes the interior of
           $\I_T$;
\item[(ii)] the conditional moment generating function of $X_t$ under $\P_x$
            has exponentially-affine dependence on $x$; that is, there exist
            functions $\phi_t(u):[0,T]\times\I_T\to\R$ and
            $\psi_t(u):[0,T]\times\I_T\to\R^d$ such that
\begin{align}\label{affine-def}
\E_x\big[\exp\langle u,X_t\rangle\big]
 = \exp\big( \phi_t(u) + \langle\psi_t(u),x\rangle \big),
\end{align}
for all $(t,u,x) \in [0,T] \times \I_T \times D$.
\end{itemize}
\end{asa}
\noindent Here $\langle\cdot,\cdot\rangle$ denotes the inner product on $\R^d$,
and $\E_x$ the expectation with respect to $\P_x$. The filtration $\bF$ is the
completed natural filtration of $X$.

The functions $\phi$ and $\psi$ satisfy the so-called \textit{generalized
Riccati equations}
\begin{subequations}\label{Riccati}
\begin{align}
\frac{\partial}{\partial t}\phi_t(u)
 &= F(\psi_t(u)),  \qquad \phi_0(u)=0, \label{Ric-1}\\
\frac{\partial}{\partial t}\psi_t(u)
 &= R(\psi_t(u)),  \qquad \psi_0(u)=u, \label{Ric-2}
\end{align}
\end{subequations}
for $(t,u)\in [0,T] \times \I_T$. The functions $F$ and $R$ are of \lk form,
that is
\begin{subequations}\label{F-R-def}
\begin{align}
F(u) &= \langle b,u\rangle +
     \int_D\big(\e^{\langle\xi,u\rangle}-1\rangle\big)m(\ud \xi),\\
R_i(u) &= \langle \beta_i,u\rangle
       + \Big\langle\frac{\alpha_i}2u,u\Big\rangle
       + \int_D\big(\e^{\langle\xi,u\rangle}-1-\langle
          u,h^i(\xi)\rangle\big)\mu_i(\ud \xi),
\end{align}
\end{subequations}
where $(b,m,\alpha_i,\beta_i,\mu_i)_{1\le i\le d}$ are \textit{admissible
parameters} and $h^i:\Rp^d\to\R^d$ are suitable truncation functions. We refer
to \cite{DuffieFilipovicSchachermayer03} for all the details.

We will later make use of the following results and definitions; here
inequalities involving vectors are interpreted component-wise.

\begin{lemma}\label{positivity}
The functions $\phi$ and $\psi$ satisfy the following:
\begin{enumerate}
\item $\phi_t(0) = \psi_t(0) = 0$ for all $t \in [0,T]$.
\item $\I_T$ is a convex set. Moreover, for each $t \in [0,T]$,
 the functions $\I_T \ni u \mapsto \phi_t(u)$ and
 $\I_T \ni u \mapsto \psi_t(u)$ are (componentwise) convex.
\item $\phi_t(\cdot)$ and $\psi_t(\cdot)$ are order-preserving: let
 $(t,u),(t,v)\in [0,T] \times \I_T$, with $u\le v$.
 Then
\begin{align}\label{order}
\phi_t(u)\le \phi_t(v)
 \quad\text{ and }\quad
\psi_t(u)\le \psi_t(v).
\end{align}
\item $\psi_t(\cdot)$ is \emph{strictly} order-preserving: let
 $(t,u),(t,v) \in [0,T] \times \I_T^\circ$, with $u < v$. Then
 $\psi_t(u) < \psi_t(v)$.
\end{enumerate}
\end{lemma}
\begin{definition}\label{def-gamma}
For any process \prozess[X] satisfying Assumption $(\mathbf{A})$,
define
\begin{align}
\gamma_X \,
 := \sup_{u\in\I_T\cap\R^d_{>0}} \E_{\mathbf{1}}\big[\e^{\scal{u}{X_T}}\big].
\end{align}
\end{definition}

Many results on affine processes can be extended to the time-inhomo\-ge\-ne\-ous
case, see \cite{Filipovic05}. The conditional moment generating function then
takes the form
\begin{align}\label{affine-conditional-inhomogeneous}
\E_x\left[\left.\exp\langle u,X_{r}\rangle\right|\F_s \right]
 = \exp\big( \phi_{s,r}(u) + \langle\psi_{s,r}(u),X_s\rangle \big),
\end{align}
for all $(s,r,u)$ such that $0 \le s \le r \le T$ and $u \in \I_T$, with
$\phi_{s,r}(u)$ and $\psi_{s,r}(u)$ now depending on both $s$ and $r$. Assuming
that $X$ satisfies the `strong regularity condition' (cf.~\cite{Filipovic05},
Definition~2.9), $\phi_{s,r}(u)$ and $\psi_{s,r}(u)$ satisfy generalized Riccati
equations with time-dependent right-hand sides:
\begin{align}
- \frac{\partial}{\partial s}\phi_{s,r}(u)
 &= F(s,\psi_{s,r}(u)),  \qquad \phi_{r,r}(u)=0, \label{Ric-TI-1}\\
- \frac{\partial}{\partial s}\psi_{s,r}(u)
 &= R(s,\psi_{s,r}(u)),  \qquad \psi_{r,r}(u)=u, \label{Ric-TI-2}
\end{align}
for all $0 \le s \le r \le T$ and $u \in \I_T$.

We close this section with an example of an affine process on $\Rp$ that has
already been used in the credit risk \cite{Duffie_Garleanu_2001} and term
structure modeling literature \cite{Filipovic_2001}.

\begin{example}
Let $X$ be a Cox--Ingersoll--Ross process with jumps, defined by the SDE
\begin{align*}
\ud X_t & = -\lambda (X_t - \theta) \dt + 2 \eta \sqrt{X_t} \ud W_t
          + \ud Z_t, \qquad X_0=x \in \R_{\geq 0},
\end{align*}
where $\lambda, \theta,  \eta \in \Rp$, $W$ is a Brownian motion and $Z$ is a
compound Poisson process with constant intensity $\ell$ and exponentially
distributed jumps with mean $\mu$. This is an affine process on $\Rp$ with
\begin{align*}
F(u) & = \lambda \theta u + \ell \frac{u}{\frac{1}{\mu} - u}, \\
R(u) & = 2 \eta^2 u^2 - \lambda u.
\end{align*}
The Riccati equations \eqref{Riccati} can be solved explicitly, and we get that
\begin{align*}
\phi_t (u) & = - \frac{\lambda \theta}{2 \eta^2}
                 \log \left(1 - 2 \eta^2 b(t) u\right) \\
 & \qquad - \frac{\ell \mu}{\lambda \mu - 2 \eta^2}  \log \left(
   \frac{\lambda (\mu u - 1)}{a(t) (\lambda \mu - 2 \eta^2) u -
         \lambda + 2 \eta^2 u} \right), \\
\psi_t(u) & = \frac{a(t) u}{1 - 2 \eta^2 b(t) u},
\end{align*}
with
\begin{align*}
a(t) & = \e^{-\lambda t} \qquad \textrm{and} \qquad
b(t) = \frac{1-\e^{-\lambda t }}{\lambda}.
\end{align*}
\end{example}

\subsection{Ordered Martingales $\geqslant1$}

The construction of the \alm is based on families of (parametrized) martingales
greater than one, which are increasing in some parameter. We follow here the
presentation of \cite{Papapantoleon10b}.

Let \prozess[X] be an affine process on $\Rp^d$ as described in the previous
section, where from now on we restrict ourselves to a finite time horizon, i.e.
$T<\infty$. Let $u\in\Rp^d$ and consider the random variable
$Y^u_T:=\e^{\scal{u}{X_T}}$. Then, from the tower property of conditional
expectations we know immediately that \prozess[M^u], where
\begin{align}\label{Pn-martingales}
M^u_t &= \E\big[\e^{\scal{u}{X_T}}|\F_t\big] \nonumber\\
      &= \exp\big( \phi_{T-t}(u) + \scal{\psi_{T-t}(u)}{X_t} \big),
\end{align}
is a martingale. Moreover, it is obvious that $M^u_t\ge1$ for all $t\in[0,T]$,
while the ordering
\begin{align}\label{M-order}
 u\le v \Longrightarrow M_t^u\le M_t^v,
\qquad \forall t\in[0,T],
\end{align}
also follows directly.

\subsection{Affine \lib models}
\label{new_approach}

Consider a discrete tenor structure $\mathcal{T}=\{0=T_0<T_1<\cdots<T_N\le T\}$
and let $\delta_k:=T_{k+1}-T_k$ for all $k\in\mathcal{K}\setminus\{N\}$, where
$\mathcal{K}:=\set{1,\dots,N}$. Let $B(\cdot,T_k)$ denote the price of a zero
coupon bond with maturity $T_k$ and $L(\cdot,T_k)$ denote the forward \lib rate
settled at $T_k$ and exchanged at $T_{k+1}$. They are related via
\begin{align}
L(t,T_k)
 = \frac{1}{\delta_k}\left(\frac{B(t,T_k)}{B(t,T_{k+1})}-1\right).
\end{align}
Denote by $\P_k$ the forward measure associated with the maturity $T_k$, i.e.
the bond $B(\cdot,T_k)$ is the numeraire, for all $k\in\mathcal{K}$. Assume that
$X$ is an affine process under $\P_N$ satisfying Assumption (\textbf{A}).

We know that discounted prices of traded assets (e.g. bonds) should be martingales with
respect to the terminal martingale measure, i.e.
\begin{align}
 \frac{B(\cdot,T_k)}{B(\cdot,T_N)} \in \mathcal{M}(\P_N),
 \qquad\text{for all}\;\, k\in\mathcal{K}.
\end{align}
The affine LIBOR ansatz is thus to model quotients of bond prices using the
$\P_N$-martingales $M^u$ as follows:
\begin{align}
\frac{B(t,T_1)}{B(t,T_N)} &= M_t^{u_1} \label{mg-2}\\
 &\vdots\nonumber\\
\frac{B(t,T_{N-1})}{B(t,T_N)} &= M_t^{u_{N-1}}, \label{mg-N}
\end{align}
for all $t\in[0,T_1],\dots,t\in[0,T_{N-1}]$ respectively, while the initial
values of the martingales $M^{u_k}$ must satisfy:
\begin{align}\label{initial-cond}
M^{u_k}_0
 = \exp\big(\phi_T(u_k) + \big\langle\psi_T(u_k),x\big\rangle\big)
 = \frac{B(0,T_k)}{B(0,T_N)},
\end{align}
for all $k\in\mathcal{K}$. Obviously we set $u_N=0\Leftrightarrow
M_0^{u_N}=\frac{B(0,T_N)}{B(0,T_N)}=1$.

We can show that under mild conditions on the underlying process $X$, an affine
LIBOR model can fit any given term structure of initial LIBOR rates through the
parameters $u_1,\dots,u_N$.

\begin{proposition}\label{initial-fit}
Suppose that $L(0,T_1), \dotsc, L(0,T_{N-1})$ is a tenor structure of
non-negative initial LIBOR rates, and let $X$ be a process satisfying Assumption
$(\mathbf{A})$, starting at the canonical value $\mathbf{1}$. The following
hold:
\begin{enumerate}
\item If $\gamma_X>B(0,T_1)/B(0,T_N)$, then there exists
 a decreasing sequence \linebreak
 $u_1\ge u_2\ge\dots\ge u_N =0$ in $\I_T\cap\Rp^d$, such that
 \begin{equation}\label{martingale-initial}
  M_0^{u_k} = \frac{B(0,T_k)}{B(0,T_N)},
   \qquad \text{for all}\;\, k\in\mathcal{K}.
 \end{equation}
 In particular, if $\gamma_X=\infty$, then the affine LIBOR model
 can fit \emph{any} term structure of non-negative initial LIBOR
 rates.
\item If $X$ is one-dimensional, the sequence
 $(u_k)_{k\in\mathcal{K}}$ is \emph{unique}.
\item If all initial LIBOR rates are positive, the sequence
 $(u_k)_{k\in\mathcal{K}}$ is \emph{strictly} decreasing.
\end{enumerate}
\end{proposition}

In this model, forward prices have the following form:
\begin{align}\label{mg-3}
1+\delta_k L(t,T_k)
 &= \frac{B(t,T_k)}{B(t,\tk[k+1])}
  = \frac{M_t^{u_k}}{M_t^{u_{k+1}}} \nonumber\\
 &= \exp\Big(\phi_{T_N-t}(u_k)- \phi_{T_N-t}(u_{k+1})\nonumber\\
 &\qquad\qquad
  + \big\langle\psi_{T_N-t}(u_k)-\psi_{T_N-t}(u_{k+1}),X_t\big\rangle\Big)
    \nonumber\\
 &= \exp\Big(A_{T_N-t}(u_k,u_{k+1})
     + \big\langle B_{T_N-t}(u_k,u_{k+1}),X_t\big\rangle\Big),
\end{align}
where we have defined
\begin{align}\label{def-A}
 A_{T_N-t}(u_k,u_{k+1})
 := \phi_{T_N-t}(u_k)- \phi_{T_N-t}(u_{k+1}),\\\label{def-B}
 B_{T_N-t}(u_k,u_{k+1})
 := \psi_{T_N-t}(u_k)- \psi_{T_N-t}(u_{k+1}).
\end{align}
Using Proposition \ref{initial-fit}(1) and Lemma \ref{positivity}(3), we
immediately deduce that \textit{\lib rates are always non-negative} in the
affine LIBOR models.

\begin{proposition}\label{libor-positiv}
Suppose that $L(0,T_1), \dotsc, L(0,T_{N-1})$ is a tenor structure of
non-negative initial LIBOR rates, and let $X$ be a process satisfying Assumption
$(\mathbf{A})$. Let the bond prices be modelled by \eqref{mg-2}--\eqref{mg-N}
and satisfying the initial conditions \eqref{initial-cond}. Then the LIBOR rates
$L(t,T_k)$ are \emph{non-negative} a.s., \emph{for all} $t\in[0,T_k]$ and
$k \in \mathcal{K}\setminus\{N\}$.
\end{proposition}

Forward measures in the \alm are related to each other via quotients of the
martingales $M^u$; we have that
\begin{align}\label{Pk-to-next}
\frac{\dd \P_k}{\dd \P_{k+1}} \Big|_{\F_t}
  = \frac{B(0,\tk[k+1])}{B(0,\tk[k])} \cdot \frac{M_t^{u_{k}}}{M_t^{u_{k+1}}}
\end{align}
for any $k\in\mathcal{K}\setminus\{N\}, t\in[0,\tk[k]]$, where each $\P_k$ is
defined on $(\Omega, \cF_{T_k})$. In addition, the
density between the $\P_k$-forward measure and the terminal forward measure
$\P_N$ is given by the martingale $M^{u_k}$, as the defining equations
\eqref{mg-2}--\eqref{mg-N} clearly dictate; we have
\begin{align}\label{Pk-to-final}
\frac{\dd \P_k}{\dd \P_N} \Big|_{\F_t}
  = \frac{B(0,\tk[N])}{B(0,\tk[k])} \cdot \frac{B(t,\tk[k])}{B(t,\tk[N])}
  = \frac{M_t^{u_k}}{M_0^{u_k}}.
\end{align}
The connection between the terminal and forward measures yields also the
\textit{martingale property of forward LIBOR rates}. We have that
\begin{align}
1+\delta_k L(\cdot,T_k)
 = \frac{M^{u_k}}{M^{u_{k+1}}} \in \mathcal{M}(\P_{k+1})
\intertext{because}
\frac{M^{u_k}}{M^{u_{k+1}}} \cdot \frac{\dd\P_{k+1}}{\dd\P_N}\Big|_{\F_\cdot}
 = M^{u_k} \in \mathcal{M}(\P_N).
\end{align}

Finally, we want to show that the model structure is preserved under any forward
measure. Indeed, we calculate the conditional moment generating function
of $X_r$ under the forward measure $\P_k$, and get
\begin{align}\label{Pk-mgf}
&\E_k\big[\e^{\langle v,X_r\rangle}\big|\F_s\big] \nonumber\\
 &= \exp\Big(\phi_{r-s}(\psi_{T_N-r}(u_{k})+ v)
      - \phi_{r-s}(\psi_{T_N-r}(u_k))\nonumber\\
 &\qquad\qquad
      + \scal{\psi_{r-s}(\psi_{T_N-r}(u_{k})+ v)
      - \psi_{r-s}(\psi_{T_N-r}(u_k))}{X_s}\Big),
\end{align}
which yields that $X$ is a time-inhomogeneous \textit{affine process under any
forward measure} $\pt[k]$, for any $k\in\mathcal{K}$. In particular, setting
$s=0$, $r=t$, we get that
\begin{equation}\label{Pk-mgf-2}
\E_k\big[\e^{\langle v,X_{t}\rangle}\big] =
\exp\left(\phi^k_t(v) + \scal{\psi^k_t(v)}{x}\right)\;,
\end{equation}
where \begin{align}
\phi^k_t(v) &:= \phi_t(\psi_{T_N-t}(u_{k})+ v) - \phi_t(\psi_{T_N-t}(u_k)),
\label{Pk-mgf-3}\\
\psi^k_t(v) &:= \psi_t(\psi_{T_N-t}(u_{k})+ v) - \psi_t(\psi_{T_N-t}(u_k)).
\label{Pk-mgf-4}
\end{align}
This also shows that the measure change from $\pt[k]$ to $\pt$ is an exponential
tilting (or Esscher transformation).

The main advantage of this modeling framework is that the affine structure of
the driving process is preserved under any forward measure, which leads to
semi-analytical pricing formulas for caps and swaptions using Fourier
transforms. Moreover, in certain examples such as the CIR model, closed-form
solutions similar to the Black--Scholes formula can be derived for both caps and
swaptions; we refer to \cite{KellerResselPapapantoleonTeichmann09} for all
the details.

\section{The defaultable \alm}
\label{dalm-sec}

In this section, we enlarge the market by adding \textit{defaultable} bonds with
zero recovery and maturities $T_k\in\mathcal{T}$. These are corporate bonds, and
default risk means the risk of default of the corporation that issued the bond.
The promised payoff at maturity $T_k$ of such a bond is one currency unit, which
is received by the bondholder (thereafter: she) if default does not occur before
or at maturity. In case of default, she receives only a partial amount of the
promised payment depending on the recovery scheme that applies. Zero recovery
means that in case of default she receives zero at maturity.

We denote the default time by $\tau$. The time-$t$ price of a defaultable bond
with zero recovery, denoted by $B^{0}(t, T_{k})$, can be written as
\begin{align}\label{dalm-basic}
B^{0}(t, T_{k}) = \ifnodefault{t} \overline{B}(t, T_{k}),
\end{align}
where $\overline{B}(t,T_{k})$, $t\in[0,T_k]$, denotes the pre-default value of
the bond, which satisfies $\overline{B}(t,T_{k}) > 0$ and
$\overline{B}(T_{k},T_{k})=1$.

Let us introduce now a concept of defaultable forward LIBOR rates, by a
straightforward generalization of the definition of the forward LIBOR rate and
using pre-default bond prices instead of default-free ones. The following
definitions are taken from \cite{EberleinKlugeSchoenbucher06}; see also
\cite{Schoenbucher00}. For a detailed discussion of similar concepts and related
defaultable FRAs we refer to \cite[\S 14.1.4, p. 431]{BieleckiRutkowski02}.

\begin{definition}
We define the \emph{defaultable forward LIBOR rate} for the period
$[T_{k},T_{k+1}]$ prevailing at time $t\leq T_{k}$ by
setting
\begin{align}
\overline{L}(t, T_{k})
 := \frac{1}{\delta_{k}} \left(
     \frac{\overline{B}(t, T_{k})}{\overline{B}(t,T_{k+1})} - 1\right).
\end{align}
\end{definition}

\begin{definition}
The \emph{forward credit spreads} between default-free and defaultable LIBOR
rates are denoted by
\begin{align}\label{libor-spread}
S(t, T_{k}) := \overline{L}(t, T_{k})-L(t, T_{k}),
\end{align}
while the associated \emph{forward default intensities} are defined as
\begin{align}
\label{H-def}
H(t, T_{k}) := \frac{\overline{L}(t, T_{k})- L(t, T_{k})}
                    {1 + \delta_{k} L(t,T_{k})}, \qquad t \leq T_{k},
\end{align}
with the convention $H(t,T_{k})= H(T_{k},T_{k})$, for $t>T_{k}$.
\end{definition}

Consequently, in terms of bond prices the following holds:
\begin{align}
\label{H-via-bonds}
H(t, T_{k})
 &= \frac{1}{\delta_{k}} \left( \frac{\overline{B}(t, T_{k})}{B(t, T_{k})}
    \frac{B(t, T_{k+1})}{\overline{B}(t, T_{k+1})} - 1\right) \nonumber\\
\Leftrightarrow
1+\delta_k H(t, T_{k})
 &= \frac{\overline{B}(t, T_{k})}{\overline{B}(t, T_{k+1})}
    \cdot \frac{B(t, T_{k+1})}{B(t, T_{k})}.
\end{align}
Each defaultable LIBOR rate can be expressed via the default-free LIBOR rate
with the same tenor date and the corresponding default intensity as
\begin{align}\label{H-df-df}
1 + \delta_{k} \overline{L}(t, T_{k})
 &= (1 + \delta_{k} L(t,T_{k})) (1 + \delta_{k} H(t, T_{k})) \nonumber\\
\Leftrightarrow
1 + \delta_{k} H(t, T_{k})
 &= \frac{1 + \delta_{k} \overline{L}(t, T_{k})}{1 + \delta_{k} L(t,T_{k})}.
\end{align}

The aim of this work is to construct an analytically tractable framework for the
joint evolution of default-free and defaultable \lib rates, where the
requirement that riskier rates are higher than risk-free ones is respected, that
is
\begin{align}
L(t,T_k) \leq \overline{L}(t,T_k)
\qquad \forall t\in[0,T_k],\, \forall k\in\mathcal{K}\setminus\{N\}.
\end{align}
In order to fulfill the last requirement, we will follow the approach in
\cite{EberleinKlugeSchoenbucher06} and model default-free \lib rates and
forward credit spreads, or equivalently forward default intensities, as
non-negative processes. In order to have an analytically tractable framework, we
will extend the \alm to the defaultable setting, i.e. we will model
\begin{align}
 1+\delta_k H(\cdot,T_k)
\end{align}
such that: (i) it remains greater than one for all times, (ii) the model is free
of arbitrage and (iii) the dynamics are of exponential-affine form.

\subsection{The Cox construction of the default time}
\label{Cox}

Here we describe the classical Cox process construction of the default time,
which is modeled as the random time when an $\bF$-adapted process crosses an
independent trigger. This construction is also known as the canonical
construction and provides a very simple and intuitive method to define the
default event. It is widely used in credit risk modeling and the details can be
found in many sources; we refer to \cite{JeanblancRutkowski00},
\cite{BieleckiRutkowski02} and \cite{EberleinKlugeSchoenbucher06}.

Let $(\Omega,\cF,\P_N)$ be a complete probability space such that a process $X$
satisfying Assumption $(\mathbf{A})$ is defined on it. Let $\bF$ denote the
completed natural filtration of $X$, and assume that $\eta$ is a random variable
defined on $(\Omega,\cF,\P_N)$, independent of $\bF$ and exponentially
distributed with mean 1. Finally, let $\Gamma$ be an $\bF$-adapted,
right-continuous, non-decreasing process such that $\Gamma_{0}=0$ and
$\lim_{t\to\infty}\Gamma_{t}= \infty$.

\begin{remark}
Note that in order to define $\Gamma$ and $\eta$ with these properties one
typically begins with a probability space where $X$ is defined and then
considers an enlarged space, obtained as the product space of the underlying
probability space and another space supporting $\eta$. Here $(\Omega,\cF,\P_N)$
is assumed to be already large enough to support the random variable $\eta$
which is independent of $\bF$.
\end{remark}

Define a random time $\tau: \Omega \to \Rp$ by
$$
\tau:= \inf \set{t \in \Rp: \Gamma_{t} \geq \eta }.
$$
This random time is not an $\bF$-stopping time. Let us denote by
$\cD_{t}:=\sigma(\ifdefault{t}: t \geq 0)$ and set $\bD=(\cD_{t})_{t\geq 0}$.
Define the filtration $\bG = (\cG_{t})_{t\geq 0}$ by setting
$\cG_{t}:=\bigcap_{s>t} (\cF_{s} \vee \cD_{s})$. Obviously, the random time
$\tau$ is a $\bG$-stopping time.

The following property can be easily proved: for all $0 \leq s \leq T_{N}$
\begin{equation}
\label{hazard_process}
\P_N (\tau > s | \cF_{T_{N}})
 = \P_N (\tau > s | \cF_{s}) = \e^{- \Gamma_{s}}.
\end{equation}
Hence, the process $\Gamma$ is by definition the $\bF$-\emph{hazard process}
of the random time $\tau$. Moreover, \eqref{hazard_process} entails the
so-called $\cH$\textit{-hypothesis}, also known as the immersion property,
namely:
\begin{itemize}
\item[($\cH$)]  Every $ \bF$-local martingale is a $ \bG$-local martingale,
\end{itemize}
which is equivalent to the following statements (cf. \cite{BremaudYor78}):
\begin{itemize}
\item[($\cH$1)] For any $t$, the $\sigma$-fields $\cF_{T_{N}}$ and $\cG_{t}$ are
conditionally independent given $\cF_{t}$, i.e.
$$
\E_N[X Y_{t} | \cF_{t}] = \E_N[X | \cF_{t}] \E_N[Y_{t} | \cF_{t}],
$$
for any bounded $\cF_{T_{N}}$-measurable $X$  and  bounded $\cG_{t}$-measurable
$Y_{t}$.
\item[($\cH$2)] For every bounded  $\cF_{T_{N}}$-measurable $X$
$$
\E_N[X | \cG_{t}]= \E_N[ X | \cF_{t}].
$$
\end{itemize}

In the sequel, we shall use the following lemma  which provides an expression
for the conditional expectation with respect to the enlarged $\sigma$-algebras
$\cG_{s}$ in terms of $\cF_{s}$. The result is classical and can be found
e.g. in \cite{JeanblancRutkowski00} or \cite{BieleckiRutkowski02}.

\begin{lemma}
\label{cond_expectation}
Let $Y$ be an integrable, $\cF$-measurable random variable. Then for any $s \leq
t$
$$
\E_N[\ifnodefault{t} Y | \cG_{s}]
= \ifnodefault{s} \frac{\E_N[\ifnodefault{t} Y|\cF_{s}]}{\E_N[\ifnodefault{t}|
\cF_{s}]}.
$$
\end{lemma}

We conclude this subsection with the following important remark.

\begin{remark}
In LIBOR modeling we consider the whole set of equivalent forward measures. Each
$\P_k$, $k \in \cK \setminus{\{N\}}$, was defined on $(\Omega, \cF_{T_k})$ via
\eqref{Pk-to-final}. We now extend this definition to the $\sigma$-algebra
$\cG_{T_k}$ using the same Radon-Nikodym derivative
$$
\frac{\dd \P_k}{\dd \P_N} = \frac{M_{T_k}^{u_k}}{M_0^{u_k}}.
$$
Using the $\mathcal{H}$-hypothesis, more precisely ($\cH$2), we have
$$
\frac{\dd \P_k}{\dd \P_N} \Big|_{\cG_t} = \frac{\dd \P_k}{\dd \P_N}
\Big|_{\cF_t} =
\frac{M_t^{u_k}}{M_0^{u_k}}.
$$
Moreover, it easily follows that $\Gamma$
is the $\bF$-hazard process of $\tau$ under \textit{all} measures $\P_k$,
$k\in\mathcal{K}$. Applying the abstract Bayes' rule and ($\cH$1) we obtain
\begin{align}\label{imp-rem}
\P_k(\tau > s | \cF_{s})
 &= \frac{ \E_N[ M^{u_{k}}_{T_{k}} \ifnodefault{s}| \cF_{s}] }{M^{u_{k}}_{s}}
  \nonumber\\\nonumber
 &= \frac{\E_N [M^{u_{k}}_{T_{k}}| \cF_{s}] \, \P_N(\tau > s | \cF_{s}) }
         {M^{u_{k}}_{s}} \\
 &= \P_N(\tau > s | \cF_{s}) = \e^{- \Gamma_{s}}.
\end{align}
\end{remark}

\subsection{No-arbitrage conditions: interplay between $H$ and $\tau$}

Before proceeding with the construction of the defaultable \alm, it is crucial
to realize that we cannot choose $H$ and $\tau$ arbitrarily. Here we follow the
argumentation of \cite{EberleinKlugeSchoenbucher06} closely; compare also
Section 4.2 in \cite{Grbac10} and in particular Proposition 4.4, Lemma 4.5 and
Remarks 4.3 and 4.6 for a discussion on the absence of arbitrage in this
framework.

Let us begin by inspecting the relationship between $\tau$ and $H$ that is
necessarily satisfied in an arbitrage-free defaultable model.

\begin{lemma}
\label{lemma:H-and-tau}
Let $H(\cdot,T_k)$, $k\in\mathcal{K}\setminus\{N\}$, be the forward default
intensities and $\tau$ the time of default. Then, in an arbitrage-free model we
have
\begin{equation}\label{H-and-tau}
1+ \delta_{k} H(t, T_{k})
 = \frac{\P_k ( \tau > T_k | \cF_{t} )}{\P_{k+1} ( \tau > T_{k+1} | \cF_{t} )}
 = \frac{\E_k [ \e^{-\Gamma_{T_{k}}}  | \cF_{t} ]}
        {\E_{k+1} [ \e^{-\Gamma_{T_{k+1}}} | \cF_{t} ]}.
\end{equation}
\end{lemma}
\begin{proof}
On the one hand, the value of a defaultable bond at maturity is
$$
B^{0}(T_{k}, T_{k})
 = \ifnodefault{T_k} \overline{B}(T_{k}, T_{k})
 = \ifnodefault{T_k}.
$$
The time-$t$ price of a contingent claim with payoff $\ifnodefault{T_k}$ at
$T_{k}$, which we denote by $\pi_{t}(\ifnodefault{T_k})$, is given by the
risk-neutral valuation formula under the forward measure $\P_k$, i.e.
\begin{align}\label{def-0}
\pi_{t}(\ifnodefault{T_k})
 = B(t, T_{k}) \, \E_k [\ifnodefault{T_k} | \cG_{t}].
\end{align}
On the other hand, $B^{0}(t,T_{k})$ denotes the time-$t$ price of a defaultable
bond. Hence, in order to have a consistent and arbitrage-free model, it
should hold
\begin{align}\label{def-1}
B^{0}(t, T_{k})
 = \pi_{t}(\ifnodefault{T_k}).
\end{align}
Now, \eqref{dalm-basic}, \eqref{def-0}, \eqref{def-1} and Lemma
\ref{cond_expectation} yield the following equality
$$
\ifnodefault{t} \overline{B}(t, T_{k})
 = \ifnodefault{t} B(t, T_{k})
   \frac{\P_k(\tau > T_k | \cF_{t})}{\P_k(\tau > t | \cF_{t} )},
$$
and from \eqref{hazard_process} and \eqref{imp-rem}, we obtain
\begin{equation*}
\label{eq:forward-def-price}
\frac{\overline{B}(t, T_{k})}{B(t, T_{k})}
 = \e^{\Gamma_{t}} \P_k(\tau > T_k | \cF_{t})
\end{equation*}
on the set $\set{\tau > t}$, for every $k \in \cK$. Recalling
\eqref{H-via-bonds} yields the first equality in \eqref{H-and-tau}. Moreover,
using the tower property of conditional expectations and the properties of the
hazard process, we get
\begin{equation*}
\E_k [ \ifnodefault{T_k} | \cF_{t}]
  = \E_k [ \E_k [\ifnodefault{T_k} | \cF_{T_{k}}] | \cF_{t}]
  = \E_k [\e^{-\Gamma_{T_{k}}} | \cF_{t}],
\end{equation*}
which combined with the first equality in \eqref{H-and-tau} yields the second
one.
\end{proof}

Therefore, as soon as the default time $\tau$ is specified, equality
\eqref{H-and-tau} produces a formula for $H$ and vice versa. In the spirit of
\cite{EberleinKlugeSchoenbucher06}, we are going to ``reverse engineer'' the
problem; that is, we shall specify the processes $H(\cdot,T_k)$,
$k\in\mathcal{K}\setminus \{N\}$, satisfying certain conditions and then define
an $\bF$-adapted process $\Gamma$ such that the relation between $H$ and
$\Gamma$ given in \eqref{H-and-tau} is satisfied. Finally, using the Cox
construction, we know that a default time $\tau$ with $\bF$-hazard process
$\Gamma$ exists.

\begin{proposition}
\label{arbitrage-free-def-model}
Assume that the default intensities $H(\cdot,T_k)$, $k\in\mathcal{K} \setminus
\{N\}$, satisfy the following assumption:
\begin{equation}
\label{H-martingale}
\bigg( \prod_{i=0}^{k} \frac{1}{1+\delta_{i} H(t, T_{i})} \bigg)_{0 \leq t \leq
T_{k}} \in \cM(\P_{k+1})
\end{equation}
for every $k\in\mathcal{K} \setminus \{N\}$, with $H(t,T_{i})=H(T_{i},T_{i})$,
for $t > T_{i}$. Moreover, let $\Gamma$ be any $\bF$-adapted, right-continuous
and non-decreasing process such that
\begin{equation}
\label{Gamma-via-H}
\Gamma_{T_{k+1}}= \sum_{i=0}^{k} \ln (1+\delta_{i} H(T_{i}, T_{i})),
\end{equation}
for every $k=0, 1, \ldots, N-1$. Then equation \eqref{H-and-tau} is satisfied.
\end{proposition}

\begin{remark}
Note that $\Gamma_{T_{k+1}}\in\cF_{T_{k}}$, thus using linear interpolation
between tenor dates $T_{k}$ and $T_{k+1}$ provides a suitable example for
$\Gamma$.
\end{remark}

\begin{proof}
We begin by noting that, in order to satisfy \eqref{H-and-tau}, it suffices to
specify the hazard process $\Gamma$ only at the tenor points $T_{k}$. Inserting
$t=T_{k}$ into \eqref{H-and-tau}, we get
$$
\E_{k+1} \big[\e^{-\Gamma_{T_{k+1}}} | \cF_{T_{k}}\big]
 = \e^{-\Gamma_{T_{k}}} \frac{1}{1+ \delta_{k} H(T_{k}, T_{k})}.
$$
This motivates us to define $\Gamma$ as in \eqref{Gamma-via-H} at tenor points
$T_k$. Now it is easily checked that this combined with the martingale property
\eqref{H-martingale} yields
\begin{equation}
\label{H-and-H}
1 + \delta_{k} H(t, T_{k})
 = \frac{\E_k \left[  \prod_{i=0}^{k-1}
         \frac{1}{1+\delta_{i} H(T_{i},T_{i})} | \cF_{t}\right]}
        {\E_{k+1} \left[ \prod_{i=0}^{k}
         \frac{1}{1+\delta_{i} H(T_{i},T_{i})} | \cF_{t}\right]},
\end{equation}
which is exactly \eqref{H-and-tau}.
\end{proof}

Finally, we will later make use of the following.
\begin{definition}
We denote by
\begin{align}\label{def-IH}
\ha(\cdot, T_{k}):= \prod_{i=0}^{k} \frac{1}{1+\delta_{i} H(\cdot, T_{i})},
\ \ \  k\in\mathcal{K}\setminus \{N\}.
\end{align}
\end{definition}

\subsection{Modeling default intensities}
\label{modeling-H}

Let us now turn our attention to the joint modeling of default-free and
defaultable LIBOR rates. Any model for this evolution should satisfy some very
basic requirements, dictated by economics, (mathematical) finance and practical
applications. In particular:
\begin{itemize}
\item credit spreads should be positive;
\item the model should be arbitrage-free;
\item dynamics should be analytically tractable.
\end{itemize}
Combining these requirements with the considerations from the previous
subsections, in order to have an arbitrage-free defaultable model that produces
positive credit spreads, the processes $\ha(\cdot,T_k)$,
$k\in\mathcal{K}\setminus\{N\}$, should satisfy the following requirements :
\begin{itemize}
\item[(\textbf{A1})]
 $1 + \delta_k H(\cdot, T_k)
  = \frac{\ha(\cdot, T_{k-1})}{\ha(\cdot,T_k)} \geq 1$,
\item[(\textbf{A2})]
 $\ha(\cdot, T_{k}) \in \cM(\P_{k+1})$.
\end{itemize}
Moreover, in order to have an analytically tractable model, we will employ the
\alm and extend it to the defaultable setting.

\begin{remark}
Note that (\textbf{A1}) immediately yields that $\ha(\cdot,T_{k})$ must be a
$[0,1]$-valued process. In addition, (\textbf{A2}) is equivalent to
\begin{itemize}
\item[(\textbf{A2$'$})]
  $\ha(\cdot, T_{k}) M^{u_{k+1}} \in \cM(\P_N)$ \quad $\forall k\in\mathcal{K}
\setminus \{N\}$,
\end{itemize}
as a consequence of the relation between forward measures \eqref{Pk-to-final}
and \cite[Prop.~III.3.8]{JacodShiryaev03}.
\end{remark}

\begin{proposition}
\label{H-affine-specification}
Assume that default-free \lib rates are modeled according to the \alm. Let
$(v_k)_{k\in\mathcal{K}}$  be a family of vectors in $\R^d$ such that $v_1
\leq u_1$ and
\begin{subequations}\label{hard-ass}
\begin{align}
\label{U4.1}
 \phi_{t}(v_k) - \phi_{t}(u_k) &\ge \phi_{t}(\vk) - \phi_{t}(\uk), \\
\label{U4.2}
 \psi_{t}(v_k) - \psi_{t}(u_k) &\ge \psi_{t}(\vk) - \psi_{t}(\uk),
\end{align}
\end{subequations}
for all $t\in[0,T_k]$ and all $k\in\mathcal{K}$. Define a family of
$\P_N$-martin\-ga\-les $M^{v_{k}}$, $k\in\mathcal{K}$, by
\begin{align}
\label{M-v-k}
M^{v_{k}}_{t}= \exp \Big( \phi_{T_{N}-t}(v_{k}) +
\scal{\psi_{T_{N}-t}(v_{k})}{X_{t}}  \Big), \quad t \leq T_k,
\end{align}
and model the $\mathbb{H}$-processes by setting
\begin{align}\label{model-IH}
\ha(t, T_{k}):=\frac{M^{v_{k+1}}_{t}}{M^{u_{k+1}}_{t}}, \qquad k\in\mathcal{K}
\setminus \{N\},
\quad t \leq T_k,
\end{align}
and $\ha(t, T_{k}) = \ha(T_k, T_{k})$, for $t > T_k$. Then the family
$\ha(\cdot, T_{k})$,  $k\in\mathcal{K}\setminus\{N\}$,  satisfies requirements
(\textbf{A1}) and (\textbf{A2}).
\end{proposition}

\begin{proof}
The specification obviously satisfies condition (\textbf{A2$'$}), or
equivalently (\textbf{A2}), i.e. $\ha(\cdot, T_{k})$ is a $\P_{k+1}$-martingale.
Let us show that it fulfills (\textbf{A1}). Firstly, by inserting $t=0$ into
\eqref{U4.2} and recalling \eqref{Ric-2}, we get
$$
v_k - u_k \geq v_{k+1} - u_{k+1},
$$
for all $k\in\mathcal{K}$.  Secondly, since $v_1 \leq u_1$ by assumption, it
follows $v_k \leq u_k$, for all $k\in\mathcal{K}$. Thus, we obtain
\begin{align}\label{IH-1}
0 \le \ha(t,T_k) \le 1, \qquad \forall k, \forall t.
\end{align}
Moreover,  \eqref{U4.1} and \eqref{U4.2} yield
\begin{align}\label{dalm-ineq}
\frac{M^{v_{k}}_{t}}{M^{u_{k}}_{t}} \geq
 \frac{M^{v_{k+1}}_{t}}{M^{u_{k+1}}_{t}}
\end{align}
for all $k\in\mathcal{K}$ and $t\in[0,T_{k}]$,
which is equivalent to
\begin{align}\label{IH-0}
\ha(t, T_{k}) \geq \ha(t, T_{k+1}), \qquad \forall k, \forall t.
\end{align}
\end{proof}

\begin{remark}
Let us briefly comment on the financial interpretation of the conditions on the
families $(u_k)$ and $(v_k)$, before we proceed with discussing some properties
of the defaultable \alm. These families satisfy the following conditions:
\begin{itemize}
\item[(C1)] $u_k\ge\uk$ for all $k\in\mathcal{K} \setminus \{N\}$,
            where $u_k\in\mathcal{I}_T\cap\Rp^d$ and $u_N=0$
\item[(C2)] $v_k\ge\vk$ for all $k\in\mathcal{K} \setminus \{N\}$, where
$v_k\in\mathcal{I}_T$
\item[(C3)] $u_k\ge v_k$ for all $k\in\mathcal{K}$
\item[(C4)] the functions $\phi$ and $\psi$ satisfy
 \begin{align}
  \tag{C4.a}
   \phi_{t}(v_k) - \phi_{t}(u_k) &\ge \phi_{t}(\vk) - \phi_{t}(\uk), \\
  \tag{C4.b}
   \psi_{t}(v_k) - \psi_{t}(u_k) &\ge \psi_{t}(\vk) - \psi_{t}(\uk).
 \end{align}
\end{itemize}
Note that condition (C2), which was not stated explicitly above, follows by
combining (C4.b) for $t=0$ and (C1). The first condition ensures that
default-free LIBOR rates are \emph{non-negative}, while the second one ensures
that defaultable LIBOR rates are \emph{non-negative}, cf. \eqref{mg-3} and
\eqref{defaultable-Libor}, respectively. The third condition ensures that the
processes $\ha(\cdot,T_k)$ are $[0,1]$-valued, while the last condition ensures
that forward default intensities $H$ are \emph{non-negative} (cf.
\eqref{model-IH}), thus the spreads between default-free and defaultable \lib
rates are also non-negative. Note that (C4) ensures also that the hazard process
$\Gamma$ defined by \eqref{Gamma-via-H} is non-decreasing. The first three
conditions are automatically satisfied for any defaultable \alm by fitting the
initial term structure of default-free and defaultable rates. The last condition
has to be imposed in addition; it is automatically satisfied, for example, for
independent affine processes, see Section \ref{example}.
\end{remark}

\subsection{Properties of the model}

Next, we show that under mild conditions on the driving affine process the
defaultable \alm can fit any initial term structure of defaultable rates. This
result also shows that conditions (C2) and (C3) are automatically satisfied.

\begin{proposition}\label{default-initial-fit}
Assume that the setting of Proposition \ref{initial-fit} is in force. Suppose
that $\overline{B}(0,T_1) \geq \overline{B}(0,T_2) \geq \cdots\geq
\overline{B}(0,T_N)$ is a tenor structure of initial defaultable bond prices
such that $\overline{B}(0,T_k) \leq B(0,T_k)$ for every $k\in\mathcal{K}$, as
well as $\overline{L}(0,T_k) \geq L(0,T_k)$, i.e.
$$
\frac{\overline{B}(0,T_k)}{\overline{B}(0,T_{k+1})} \geq
\frac{B(0,T_k)}{B(0,T_{k+1})}.
$$
Let $X$ be a process satisfying Assumption $(\mathbf{A})$, starting at the
canonical value $\mathbf{1}$. The following hold:
\begin{enumerate}
\item If $\gamma_X> B(0,T_1)/B(0,T_N)$, then there exists
 a decreasing sequence
 $v_1\ge v_2 \ge \dots\ge v_N$ in $\I_T$, such that
 \begin{equation}\label{martingale-initial-v}
  M_0^{v_k} = \frac{\overline{B}(0,T_{k})}{B(0,T_{N})},
   \qquad \text{for all }\;k\in\mathcal{K}.
 \end{equation}
 In particular, if $\gamma_X=\infty$, then the defaultable affine LIBOR model
 can fit \emph{any} term structure of non-negative initial defaultable LIBOR
 rates. Moreover, for each $k\in\mathcal{K}$ it holds:
 $v_{k}\leq u_{k}$.
\item If $X$ is one-dimensional, the sequence
 $(v_k)_{k\in\mathcal{K}}$ is \emph{unique}.
\item If all initial defaultable LIBOR rates are positive, then the sequence
 $(v_k)_{k\in\mathcal{K}}$ is \emph{strictly} decreasing.
\end{enumerate}
\end{proposition}

\begin{remark}
Note that since $\overline{B}(0,T_1)\leq B(0,T_1)$ by assumption, it follows
that as soon as $\gamma_X$ satisfies condition (1) above, it will automatically
follow that $\gamma_X> \overline{B}(0,T_1)/B(0,T_N)$.
\end{remark}

\begin{proof}
We have that the tenor structure of initial defaultable bond prices satisfies
$$
\ha(0, T_k) = \prod_{i=0}^{k} \frac{1}{1+\delta_k H(0,
T_i)}=\frac{\overline{B}(0,T_{k+1})}{B(0, T_{k+1})} \leq 1.
$$
Recalling that $M^{u_{k+1}}_{0} = \frac{B(0,T_{k+1})}{B(0, T_N)}$ we obtain
$$
\ha(0, T_k) M^{u_{k+1}}_{0} = \frac{\overline{B}(0,T_{k+1})}{B(0, T_{k+1})}
\frac{B(0,T_{k+1})}{B(0, T_N)} = \frac{\overline{B}(0,T_{k+1})}{B(0, T_N)}.
$$
Note that by assumption
\begin{equation}
\label{inequalities}
\frac{\overline{B}(0,T_{1})}{B(0, T_N)} \geq \frac{\overline{B}(0,T_{2})}{B(0,
T_N)} \geq \cdots \geq \frac{\overline{B}(0,T_{N})}{B(0, T_N)} > 0,
\end{equation}
where the last term $\frac{\overline{B}(0,T_{N})}{B(0, T_N)} \leq 1$. Therefore,
similarly to the proof of Proposition \ref{initial-fit} (cf. Proposition 6.1 in
\cite{KellerResselPapapantoleonTeichmann09}), we can find a decreasing
sequence $(v_k)_{k \in \cK}$ such that
$$
M^{v_{k}}_{0} = \frac{\overline{B}(0,T_{k})}{B(0, T_N)}.
$$
More precisely, take $u_{+}$ as defined therein: let $u_{+}\in\I_{T}\cap\Rp^d$
be such that
$$
\E_{\mathbf{1}}[\e^{\langle u_{+}, X_T \rangle}] > \gamma_X - \varepsilon >
\frac{B(0,T_{1})}{B(0, T_N)},
$$
where $\varepsilon > 0$ is small enough such that $\gamma_X - \varepsilon >
\frac{B(0,T_{1})}{B(0, T_N)}$. Note that $u_{+}$ must exist by definition of
$\gamma_X$. Similarly, since $\inf_{v \in \Rm^{d}}
\E_{\mathbf{1}}[\e^{\langle v, X_T \rangle}]=0$, we can find some $\lambda<0$ such that
$$
\E_{\mathbf{1}}\Big[\e^{\langle \lambda u_{+}, X_T \rangle}\Big] <
\frac{\overline{B}(0,T_{N})}{B(0, T_N)} \leq 1.
$$
We consider the function $f$ defined in the aforementioned proposition and
extend its domain to the interval $[\lambda, 1]$, i.e. we define
$$
f:[\lambda, 1] \to \Rp, \qquad \xi \mapsto f(\xi)=\E_{\mathbf{1}}\Big[\e^{\langle \xi
u_{+}, X_T
\rangle}\Big] = M_{0}^{\xi u_{+}}.
$$
This function was already shown to be continuous and increasing. Moreover,
$f(\lambda) < \overline{B}(0,T_{N})/ {B(0, T_N)}$ and $f(1) >
\overline{B}(0,T_{1})/ {B(0, T_N)}$,  since $f(1) >B(0,T_1)/B(0,T_N)$ by
definition of $u_+$.
Thus,
there exists a sequence $\lambda<\eta_N\leq\cdots\leq\eta_1 < 1$ such that
$$
f(\eta_k)=M_{0}^{\eta_k u_{+}} = \frac{\overline{B}(0,T_{k})}{B(0, T_N)}, \quad
k \in \cK.
$$
Setting $v_k:= \eta_k u_{+}$ we obtain the desired decreasing sequence. Note
that as soon as there exists $k_{0}$ such that
$\frac{\overline{B}(0,T_{k_{0}})}{B(0, T_N)} <1$, it follows that  $v_k \in
\Rm^d$, for all $k\geq k_{0}$.

Moreover, we have that $v_k \leq u_k$, since
$\frac{\overline{B}(0,T_k)}{B(0,T_N)} \leq \frac{B(0,T_k)}{B(0, T_N)}$.

If $X$ is one-dimensional, then any choice of $u_{+}$ and $\lambda$ leads to the
same parameters $v_{k}$, which shows (2).

Finally, if the initial defaultable LIBOR rates are positive, inequalities in
\eqref{inequalities} become strict and thus the sequence $(v_k)$ becomes
strictly decreasing (see again Proposition \ref{initial-fit}).
\end{proof}

\begin{remark}
Note that from the assumption
$\frac{\overline{B}(0,T_{k})}{\overline{B}(0,T_{k+1})} \geq
\frac{B(0,T_{k})}{B(0,T_{k+1})}$, it follows directly that
\begin{align}\label{initial-data}
 \frac{M^{v_{k}}_{0}}{M^{u_{k}}_{0}} \geq
\frac{M^{v_{k+1}}_{0}}{M^{u_{k+1}}_{0}}.
\end{align}
Using \eqref{initial-cond} we get that
\begin{align}\label{rem-ini-data}
\phi_{T_N}(v_k)  - \phi_{T_N}(u_k)
&+ \big\langle\psi_{T_N}(v_k) - \psi_{T_N}(u_k),\mathbf{1}\big\rangle
\geq \nonumber \\
\phi_{T_N}(v_{k+1}) - \phi_{T_N}(u_{k+1})
&+ \big\langle\psi_{T_N}(v_{k+1}) - \psi_{T_N}(u_{k+1}),\mathbf{1}\big\rangle.
\end{align}
which agrees with \eqref{hard-ass}.
\end{remark}

Obviously the defaultable \alm inherits many properties from its default-free
counterpart: the defaultable rates are non-negative and the dynamics have an
exponential-affine structure.

\begin{lemma}
The defaultable LIBOR rate $\overline{L}(\cdot, T_{k})$ has the following form
\begin{align}\label{defaultable-Libor}
1+\delta_{k}\overline{L}(t, T_{k})
 &= \frac{M^{v_{k}}_{t}}{M^{v_{k+1}}_{t}} \\\nonumber
 &= \exp\Big(A_{T_N-t}(v_k,v_{k+1})
     + \big\langle B_{T_N-t}(v_k,v_{k+1}),X_t\big\rangle\Big) \ge1,
\end{align}
for all $T_{k}\in\mathcal{T}$ and $t \leq T_{k}$, where $A$ and $B$ are defined
in \eqref{def-A}--\eqref{def-B}.
\end{lemma}
\begin{proof}
We have that
\begin{align}\label{def-intensity-affine}
1+\delta_{k}H(t, T_{k})
& = \frac{\ha(t, T_{k-1})}{\ha(t, T_{k})}
  = \frac{M^{v_{k}}_{t}}{M^{u_{k}}_{t}} \frac{M^{u_{k+1}}_{t}}{M^{v_{k+1}}_{t}}.
\end{align}
Using \eqref{mg-3} and \eqref{H-df-df} we deduce
\begin{align*}
1+\delta_{k}\overline{L}(t, T_{k})
 &= (1+\delta_{k}L(t, T_{k}))(1+\delta_{k}H(t, T_{k}))\\
 &= \frac{M^{u_{k}}_{t}}{M^{u_{k+1}}_{t}} \cdot
    \frac{M^{v_{k}}_{t}}{M^{u_{k}}_{t}} \frac{M^{u_{k+1}}_{t}}{M^{v_{k+1}}_{t}}
  = \frac{M^{v_{k}}_{t}}{M^{v_{k+1}}_{t}},
\end{align*}
which yields that the dynamics of defaultable rates are of exponential-affine
form. Positivity follows from Lemma \ref{positivity}(3) and condition (C2).
\end{proof}

Finally, we summarize below the main properties of the defaultable affine LIBOR
model.

\begin{proposition}
Suppose the conditions of Propositions \ref{initial-fit} and
\ref{default-initial-fit} are satisfied and assume \eqref{hard-ass}. Then the
defaultable affine LIBOR model given by \eqref{mg-2}--\eqref{mg-N} and
\eqref{model-IH} with initial conditions \eqref{initial-cond} and
\eqref{martingale-initial-v} is free of arbitrage. The LIBOR rates
$L(\cdot,T_k)$ and the defaultable LIBOR rates $\overline{L}(\cdot, T_{k})$ are
non-negative a.s., for all $k \in \cK\setminus\{N\}$,  and have
exponentially-affine dynamics.
\end{proposition}
\begin{proof}
The defaultable affine LIBOR model is free of arbitrage by Propositions
\ref{arbitrage-free-def-model} and \ref{H-affine-specification}. The other
claims were already proved above.
\end{proof}

\begin{remark}
Note that equation \eqref{Gamma-via-H} implies that, in \dalms, the hazard
process $\Gamma$ is an affine transformation of the driving affine process $X$
at tenor dates $T_k$, $k\in\cK$. More precisely, we have
\begin{align}\label{affine-hazard}
\Gamma_{T_{k+1}}
 & = \ln \left( \ha(T_k, T_k)^{-1}\right)
   = \ln \left(\frac{M_{T_k}^{u_{k+1}}}{M_{T_k}^{v_{k+1}}}\right) \nonumber\\
 & = A_{T_N-T_k}(u_{k+1},v_{k+1})
     + \big\langle B_{T_N-T_k}(u_{k+1},v_{k+1}),X_{T_k}\big\rangle,
\end{align}
where $A$ and $B$ are defined in \eqref{def-A}--\eqref{def-B}. In addition, we
can embed this model in a Heath--Jarrow--Morton framework for defaultable bonds,
by extending the tenor structure to a continuous term structure. This extension
preserves the properties of the model, in particular \eqref{affine-hazard}
remains true. This provides a direct link between \dalms and intensity models
for credit risk which are driven by affine processes; see
\cite[Chapter~22]{BrigoMercurio06} for a detailed overview of intensity models.
\end{remark}

\subsection{Example: independent affine processes}
\label{example}

In this subsection, we provide an example of two families of processes
$\{M^{u_k}; k \in \cK\}$ and $\{M^{v_k}; k \in \cK \}$ which satisfy our
modeling requirements, in particular inequality \eqref{dalm-ineq}. The
construction relies on independent affine processes.

Let $d_1,d_2\in\mathbb{N}$ with $d=d_1+d_2$, where $d$ is the dimension of the
affine process $X$. The first $d_1$ and the last $d_2$ components of $X$ are
$d_1$-, respectively $d_2$-dimensional affine processes, denoted by $X^1$ and $X^2$, assuming that the
filtration $\bF$ is generated by $X^1,X^2$; see Proposition 4.8 in
\cite{KellerRessel08}. In addition, we assume that $X^1$ and $X^2$ are mutually \emph{independent}.
Then, we have the following result
\begin{eqnarray}\label{eq:indep-affine}
\E_N \big[ \e^{\langle u, X_{T_{N}\rangle }} \big| \cF_t\big]
 & = & \exp \Big(\phi_{T_N-t}(u) +
      \big\langle\psi_{T_N-t}(u),X_t\big\rangle\Big) \nonumber\\
 & = & \exp \Big(\phi^1_{T_N-t}(u^1) + \phi^2_{T_N-t}(u^2) \\\nonumber
 && \qquad\quad + \big\langle\psi^1_{T_N-t}(u^1),X^1_t\big\rangle +
       \big\langle\psi^2_{T_N-t}(u^2),X^2_t\big\rangle\Big),
\end{eqnarray}
where $\phi^i$ and $\psi^i$ correspond to $X^i$, $i=1,2$, in the sense of
Assumption (\textbf{A}), while $u=(u^1, u^2)\in\R^{d_1}\times\R^{d_2}$. See
\cite[Proposition 4.7]{KellerRessel08}.

The families of $\P_N$-martingales $\{M^{u_k}\}$ and $\{M^{v_k}\}$ are
constructed in the following way:

\textit{Step 1:} We begin by constructing martingales $\{M^{u_k}; k \in \cK\}$.
First we apply Proposition \ref{initial-fit} to the initial values of the LIBOR
rates and the driving process $X^1$. We obtain a decreasing sequence $\bar{u}_1
\geq \bar{u}_2 \geq \cdots \geq \bar{u}_N =0$, where $\bar{u}_k \in \R^{d_1}$,
for every $k \in \cK$. For each $\bar{u}_k$, let us denote $u_k:=(\bar{u}_k, 0,
\ldots, 0) \in \R^{d}$. Then we have
\begin{eqnarray*}
M_t^{u_k} & = &  \exp \Big(\phi^1_{T_N-t}(\bar{u}_k) +
\big\langle\psi^1_{T_N-t}(\bar{u}_k),X_t^1\big\rangle\Big)\\
& = & \exp \Big(\phi_{T_N-t}(u_k) +
\big\langle\psi_{T_N-t}(u_k),X_t\big\rangle\Big),
\end{eqnarray*}
where the second equality follows from \eqref{eq:indep-affine} and Lemma
\ref{positivity}(1) applied to $X^2$. The martingales $(M^{u_k})_{k\in\cK}$, are
used to model \lib rates.

\textit{Step 2:} Next, we construct the processes $\frac{M^{v_k}}{M^{u_k}}$,
$k\in\mathcal{K}$, by setting
$$
\frac{M^{v_k}_t}{M^{u_k}_t}
 =  \exp \Big(\phi^2_{T_N-t}(\bar{w}_k) +
   \big\langle\psi^2_{T_N-t}(\bar{w}_k),X_t^2\big\rangle\Big)
 =: M_t^{\bar{w}_k},
$$
where $\bar{w}_k \in \R^{d_2}$ are obtained by applying the same procedure as in
the proof of Proposition \ref{default-initial-fit} to the affine process $X^2$
and the initial values
$$
M^{\bar{w}_k}_0 =
\ha(0, T_{k-1}) = \frac{\overline{B}(0, T_k)}{B(0, T_k)} \leq 1, \qquad k\in\cK.
$$
Note that $\bar{w}_k \leq 0$ by construction. Moreover, the sequence $0 \geq
\bar{w}_1 \geq \bar{w}_2 \geq \cdots \geq \bar{w}_N$ is decreasing, which
follows from the initial conditions
$$
\frac{\overline{B}(0, T_k)}{B(0, T_k)}  \geq \frac{\overline{B}(0,
T_{k+1})}{B(0, T_{k+1})}.
$$

Consequently, applying the ordering \eqref{M-order} we directly conclude that
$$
\frac{M^{v_k}_t}{M^{u_k}_t} \geq \frac{M^{v_{k+1}}_t}{M^{u_{k+1}}_t}
\quad\text{ since }\quad
M_t^{\bar{w}_k} \ge M_t^{\bar{w}_{k+1}},
$$
for every $k\in\cK\setminus\{N\}$ and every $t\in[0,T_k]$. Hence, condition
\textbf{(A1)} is satisfied.

\textit{Step 3:} Finally, it remains to verify condition \textbf{(A2)}, which
reads  as follows: $\ha(\cdot, T_{k-1}) M^{u_{k}} = M^{v_k} \in \cM(\P_N)$. We
have
\begin{eqnarray*}
M^{v_k}_t
& = & M^{u_{k}}_t \cdot M_t^{\bar{w}_k}\\
& = & \exp \Big(\phi^1_{T_N-t}(\bar{u}_k) + \phi^2_{T_N-t}(\bar{w}_k) \\
&& \qquad \quad + \big\langle\psi^1_{T_N-t}(\bar{u}_k),X_t^1\big\rangle +
  \big\langle\psi^2_{T_N-t}(\bar{w}_k),X_t^2\big\rangle\Big) \\
& = & \exp \Big(\phi_{T_N-t}(v_k) +
  \big\langle\psi_{T_N-t}(v_k),X_t\big\rangle\Big) \\
& = &  \E_N\big[\e^{\langle v_k, {X_{T_N}} \rangle } \big| \cF_t\big],
\end{eqnarray*}
where we defined $v_k:=(\bar{u}_k, \bar{w}_k) \in \R^d$. Hence, $M^{v_k}
\in \cM(\P_N)$, for all $k \in \cK$.

\begin{remark}
The existence of \textit{dependent} affine processes that satisfy these
requirements remains an open question. More generally, the construction of
ordered martingales that satisfy inequality \eqref{dalm-ineq} seems to be
non-trivial.
\end{remark}

\section{Pricing credit derivatives}
\label{dalm-price}

The pricing of credit derivatives in the defaultable \alm is a (relatively)
simple task due to the analytical tractability of the model. In particular, we
can derive explicit expressions for derivatives with line\-ar payoffs, such as
credit default swaps, and semi-analytical formulas for products with non-linear
payoffs, utilising the affine property and Fourier methods. Our formulas do not
involve any approximation; compare with \cite{EberleinKlugeSchoenbucher06}
where approximations are necessary.

An essential tool for the pricing of credit derivatives are restricted
defaultable forward measures, introduced by \cite{Schoenbucher00} and further
exploited by \cite{EberleinKlugeSchoenbucher06}. They are the restrictions
of defaultable forward martingale measures, also called survival measures and
defined on $(\Omega,\cG_{T_{k}})$, to the sub-$\sigma$-fields $\cF_{T_{k}}$ for
each $k\in\mathcal{K}$; see \cite[Defs.~15.2.1,~15.2.2]{BieleckiRutkowski02}.

\begin{definition}
The restricted defaultable forward martingale measure $\overline{\P}_k$
associated to the maturity $T_{k}$, $k\in\mathcal{K}$, is given on
$(\Omega,\cF_{T_{k}})$ by
\begin{align}\label{def-measure}
\frac{\dd \overline{\P}_k}{\dd \P_k}\Big|_{\cF_{t}}
 = \frac{B(0, T_{k})}{\overline{B}(0, T_{k})}
   \P_k (\tau > T_{k}| \cF_{t}).
\end{align}
\end{definition}

The explicit relation between the default time $\tau$ and the forward default
intensity yields
\begin{align*}
\P_k (\tau > T_{k}| \cF_{t})
 &= \E_k \big[\e^{- \Gamma_{T_{k}}}|\cF_{t}\big]\\
 &= \E_k \big[ \ha(T_{k-1}, T_{k-1}) | \cF_{t} \big]
  = \ha(t,T_{k-1}),
\end{align*}
hence we can deduce that
\begin{equation}
\label{P-bar}
\frac{\dd \overline{\P}_k}{\dd \P_k}\Big|_{\cF_{t}}
 = \frac{B(0,T_{k})}{\overline{B}(0, T_{k})} \cdot
   \frac{M_{t}^{v_{k}}}{M_{t}^{u_{k}}}.
\end{equation}
Moreover, the density process between the restricted defaultable forward
measures is
described by
\begin{align}
\frac{\dd \overline{\P}_k}{\dd \overline{\P}_{k+1}}\Big|_{\cF_{t}}
 = \frac{\overline{B}(0, T_{k+1})}{\overline{B}(0,T_{k})} \cdot
   \frac{M_{t}^{v_{k}}}{M_{t}^{v_{k+1}}};
\end{align}
compare with expression \eqref{Pk-to-next} for the default-free forward
measures.

These results clarify some important properties of the defaultable \alm. On the
one hand, it easily follows from \eqref{defaultable-Libor} that the defaultable
LIBOR rate $\overline{L}(\cdot,T_{k})$ is a $\overline{\P}_{k+1} $-martingale.
On the other hand, we can deduce that the defaultable \alm remains
\textit{analytically tractable}, in the sense that the driving process preserves
the affine property under any restricted defaultable forward measure. Of course,
as in the default-free case, it becomes time-inhomogeneous. Indeed, reasoning as
in \eqref{Pk-mgf}--\eqref{Pk-mgf-4}, we get that
\begin{align}\label{Pbk-mgf-2}
\overline\E_k\big[\e^{\langle w,X_{t}\rangle}\big]
 &= \E_N\bigg[\e^{\langle w,X_{t}\rangle} \cdot \frac{\ud\overline\P_k}{\ud\P_k}
      \frac{\ud\P_k}{\ud\P_N}\Big|_{\F_t}\bigg]
  = \E_N\bigg[\e^{\langle w,X_{t} \rangle}
      \frac{M_t^{v_k}}{M_0^{v_k}} \bigg] \nonumber\\
 &= \exp\left(\overline\phi^k_t(w) + \scal{\overline\psi^k_t(w)}{x}\right)\;,
\end{align}
where \begin{align}
\overline\phi^k_t(w)
 &:= \phi_t(\psi_{T_N-t}(v_{k})+ w) - \phi_t(\psi_{T_N-t}(v_k)),
\label{Pbk-mgf-3}\\
\overline\psi^k_t(w)
 &:= \psi_t(\psi_{T_N-t}(v_{k})+ w) - \psi_t(\psi_{T_N-t}(v_k)).
\label{Pbk-mgf-4}
\end{align}

\subsection{Credit default swaps}

\emph{Credit default swaps} are credit derivatives used to provide protection
against default of an underlying asset. Consider a maturity date $T_{m}$ and a
defaultable coupon bond with fractional recovery of treasury value as the
underlying asset. The coupons with value $c$ are promised to be paid at the
dates $T_{1}, \ldots, T_{m}$ and, in case of default before maturity, a fixed
fraction $\pi \in [0,1)$ of the notional is received by the owner of the bond.
The protection buyer in such a credit default swap pays a fixed amount
$\mathpzc{S}$ periodically at dates $T_{0}, T_{1}, \ldots, T_{m-1}$ until
default and the protection seller promises to make a payment that covers the
loss if default happens, i.e.
$$
1-\pi (1+c)
$$
is paid to the protection buyer at $T_{k+1}$ if default occurs in
$(T_{k},T_{k+1}]$,
$k \in \{0, 1, 2, \ldots, m-1\}$.

The value at time 0 of the fee payments is given by
$$
\mathpzc{S} \sum_{l=1}^{m} \overline{B}(0, T_{l-1}),
$$
and the value of the default payment is given by
\begin{eqnarray*}
\lefteqn{\sum_{k=1}^{m} \Big( B(0, T_k) \E_k \big[(1-\pi (1+c))
\left(\indik_{\{\tau>T_{k-1}\}} - \indik_{\{\tau >T_{k}\}} \right) \big]\Big)}\\
& = & (1-\pi (1+c)) \sum_{k=1}^{m} \Big( \overline{B}(0, T_{k}) \delta_{k-1}
     \overline\E_k\big[H(T_{k-1},T_{k-1})\big] \Big);
\end{eqnarray*}
see \cite[Lemma 4 and Section 6]{EberleinKlugeSchoenbucher06}. The CDS rate,
also known as the CDS spread, is defined as the level $\mathpzc{S}$ that makes
the value of the credit default swap at inception equal to zero. We have that
\begin{equation}\label{cds-1}
\mathpzc{S} = \frac{1-\pi (1+c)}{ \sum_{l=1}^{m} \overline{B}(0, T_{l-1}) }
     \sum_{k=1}^{m} \Big( \overline{B}(0, T_{k}) \delta_{k-1}
     \overline\E_k\big[H(T_{k-1},T_{k-1})\big] \Big).
\end{equation}
In the defaultable \alm the forward default intensity has an exponential affine
form, in particular we have from \eqref{def-intensity-affine}
\begin{align}\label{cds-2}
1+\delta_{k-1} H(T_{k-1},T_{k-1})
 &= \frac{M^{v_{k-1}}_{T_{k-1}} M^{u_k}_{T_{k-1}}}
         {M^{u_{k-1}}_{T_{k-1}} M^{v_k}_{T_{k-1}}}
  = \e^{A_k + B_k \cdot X_{T_{k-1}}},
\end{align}
where
\begin{align*}
A_{k} &:= \phi_{T_N-T_{k-1}}(v_{k-1}) - \phi_{T_N-T_{k-1}}(u_{k-1})
        - \phi_{T_N-T_{k-1}}(v_k) + \phi_{T_N-T_{k-1}}(u_k), \\
B_{k} &:= \psi_{T_N-T_{k-1}}(v_{k-1}) - \psi_{T_N-T_{k-1}}(u_{k-1})
        - \psi_{T_N-T_{k-1}}(v_k) + \psi_{T_N-T_{k-1}}(u_k).
\end{align*}
Using the affine property of $X$ under restricted defaultable forward measures,
we can deduce a closed-form expression for the CDS spread. We have, from
\eqref{cds-1}, \eqref{cds-2} and \eqref{Pbk-mgf-2}, that
\begin{align}
\mathpzc{S}
 &= \frac{1-\pi (1+c)}{ \sum_{k=1}^{m} \overline{B}(0, T_{k-1})} \nonumber\\
 &\quad\times \sum_{k=1}^{m} \overline{B}(0, T_{k}) \left(
 \exp\Big\{ A_k + \overline\phi^k_{T_{k-1}}(B_k)
  + \big\langle\overline\psi^k_{T_{k-1}}(B_k),X_0\big\rangle \Big\} - 1 \right).
\end{align}

\begin{remark}
Analogous closed-form expressions for other credit derivatives with linear
payoffs, such as total rate of returns swaps and asset swap packages, can be
easily derived. See \cite[\S 4.6]{Kluge05} for more details on credit
derivatives in defaultable \lib models.
\end{remark}

\begin{remark}
Note that when the processes driving the risk-free interest rates and the
default intensities are independent (in other words, when the risk-free rates
and the default time are independent), the CDS spread can be expressed as a
function of the initial default-free and defaultable bond prices and is
model-independent. This is a well-known property, discussed for example in
\cite[\S~21.3.5]{BrigoMercurio06} for the continuous tenor case. Let us show
it in our framework. We recall Example \ref{example} and first note that
\begin{equation}
\label{H-indep-case}
\overline{\E}_k \left[H(T_{k-1},T_{k-1})\right] = H(0,T_{k-1}).
\end{equation}
Then, it follows directly from \eqref{cds-1} that
\begin{align*}
\mathpzc{S} & = \frac{1-\pi (1+c)}{ \sum_{l=1}^{m} \overline{B}(0, T_{l-1}) }
     \sum_{k=1}^{m} \Big( \overline{B}(0, T_{k}) \delta_{k-1}
     H(0,T_{k-1}) \Big) \\
     & = \frac{1-\pi (1+c)}{ \sum_{l=1}^{m} \overline{B}(0, T_{l-1}) }
     \sum_{k=1}^{m}  \frac{\overline{B}(0, T_{k-1}) B(0, T_{k}) -
\overline{B}(0, T_{k}) B(0, T_{k-1}) }{B(0, T_{k-1})}.
\end{align*}
This formula can be used to bootstrap the initial defaultable bond prices from
the CDS spreads quoted in the market. In order to show \eqref{H-indep-case}, we
use the independence and the martingale property of $M^u$ and $M^{\bar{w}}$; we
have
\begin{align*}
\overline{\E}_k \left[1+\delta_{k-1} H(T_{k-1},T_{k-1})\right]
 &=   \E_N \left[ \frac{M^{v_{k-1}}_{T_{k-1}} M^{u_k}_{T_{k-1}}}
         {M^{u_{k-1}}_{T_{k-1}} M^{v_k}_{T_{k-1}}} \frac{M^{v_{k}}_{T_{k-1}}}
         {M^{v_k}_{0}}\right] \\
 &= \frac{1}{{M^{v_k}_{0}}}
    \E_N \left[ M^{\bar{w}_{k-1}}_{T_{k-1}} M^{u_k}_{T_{k-1}} \right] \\
 &= \frac{1}{{M^{v_k}_{0}}}
    \E_N [ M^{\bar{w}_{k-1}}_{T_{k-1}} ]  \E_N [M^{u_k}_{T_{k-1}} ] \\
 &= \frac{M^{\bar{w}_{k-1}}_{0} }{{M^{\bar{w}_{k}}_{0}}}
  = 1+\delta_{k-1} H(0,T_{k-1}).
\end{align*}
\end{remark}

\subsection{Options on defaultable bonds}

We consider now options on defaultable bonds, and focus on a European call on a
defaultable zero coupon bond, for simplicity. Options on defaultable fixed or
floating coupon bonds can be treated similarly. Let $T_i$ be the maturity and
$K\in(0, 1)$ the strike of a call option on a defaultable zero coupon bond with
maturity $T_m \geq T_i$. We follow \cite{Kluge05}, and adopt the fractional
recovery of treasury value scheme, which means that in case of default prior to
maturity of the bond the owner receives the amount $\pi\in (0, 1)$ at maturity
$T_m$; see \cite{BieleckiRutkowski02} for details and alternative recovery
schemes. We denote the price of this bond by $B^{\pi}(\cdot, T_m)$ and its
time-$t$ value equals
\begin{align*}
B^{\pi}(t,T_m) = \pi B(t,T_m) + (1-\pi)\indik_{\{\tau>t\}}\overline{B}(t,T_m).
\end{align*}
The payoff of the option at maturity $T_i$ is given by $\indik_{\{\tau > T_i\}}
(B^{\pi}(T_i, T_m) - K)^+$, which means that it is knocked out at default.

The price of this option, using \eqref{def-measure}, is provided by
\begin{align}\label{def-options}
\pi_0^{\text{CO}}
 & = B(0, T_i)\, \E_i [\indik_{\{\tau > T_i\}} (B^{\pi}(T_i, T_m) - K)^+]
   \nonumber\\
 & = \overline{B}(0, T_i)\, \overline{\E}_i [(\pi B(T_i, T_m)  +
     (1-\pi)\overline{B}(T_i, T_m) - K)^+] \nonumber\\
 &= \overline{B}(0,T_i)\, \overline\E_i \left[\left(
     \pi \prod_{l=i}^{m-1} (1+\delta_lL(T_i,T_l))^{-1} \right.\right.\nonumber\\
 &\qquad\qquad\qquad\qquad \left.\left.
     + (1-\pi) \prod_{l=i}^{m-1} (1+\delta_l\overline{L}(T_i,T_l))^{-1}
     - K\right)^+\right].
\end{align}
Now, in the default-free and \dalms we have that
\begin{align*}
1+\delta_lL(T_i,T_l)
 &= \frac{M_{T_i}^{u_l}}{M_{T_i}^{u_{l+1}}}
  = \exp\left( A_{i,l} + \scal{B_{i,l}}{X_{T_i}} \right), \\
1+\delta_l \overline{L}(T_i,T_l)
 &= \frac{M_{T_i}^{v_l}}{M_{T_i}^{v_{l+1}}}
  = \exp\left( \overline{A}_{i,l} + \scal{\overline{B}_{i,l}}{X_{T_i}} \right),
\end{align*}
where
\begin{align*}
 A_{i,l} &= A_{T_N-T_i}(u_l,u_{l+1}),\qquad
 B_{i,l} = B_{T_N-T_i}(u_l,u_{l+1}),\\
 \overline{A}_{i,l} &= A_{T_N-T_i}(v_l,v_{l+1}),\qquad
 \overline{B}_{i,l} = B_{T_N-T_i}(v_l,v_{l+1}).
\end{align*}
Therefore, for the product terms in \eqref{def-options} we get that
\begin{align*}
\label{product-Libor}
\prod_{l=i}^{m-1} (1+\delta_lL(T_i,T_l))^{-1}
 &= \exp\left( A^m_i + \scal{B^m_i}{X_{T_i}} \right),
\end{align*}
and
\begin{align*}
\prod_{l=i}^{m-1} (1+\delta_l\overline{L}(T_i,T_l))^{-1}
 &= \exp\left( \overline{A}^m_i
    + \scal{\overline{B}^m_i}{X_{T_i}} \right),
\end{align*}
with the obvious definitions
\begin{align}
A^m_i &= -\sum_{l=i}^{m-1} A_{i,l}, \qquad
B^m_i  = -\sum_{l=i}^{m-1} B_{i,l} \\
\overline{A}^m_i &= -\sum_{l=i}^{m-1} \overline{A}_{i,l}, \qquad
\overline{B}^m_i  = -\sum_{l=i}^{m-1} \overline{B}_{i,l}.
\end{align}
Therefore, returning to the option pricing problem, we have
\begin{align}\label{def-options-2}
\pi_0^{\text{CO}}
 &= \overline{B}(0,T_i)\, \overline\E_i \left[\left(
     \pi \e^{A^m_i + \scal{B^m_i}{X_{T_i}}}
     + (1-\pi) \e^{\overline{A}^m_i + \scal{\overline{B}^m_i}{X_{T_i}}}
     - K\right)^+\right] \nonumber\\
 &= \overline{B}(0,T_i)\, \overline\E_i \left[\left(
     \e^{Y_1} + \e^{Y_2} - K\right)^+\right],
\end{align}
where
\begin{align}
Y_1 &:= \log\pi + A^m_i + \scal{B^m_i}{X_{T_i}}, \\
Y_2 &:= \log(1-\pi) + \overline{A}^m_i
      + \scal{\overline{B}^m_i}{X_{T_i}}.
\end{align}

Now, the expression in \eqref{def-options-2} corresponds to the payoff of a
spread option
\begin{align}
 g(x_1,x_2) = (\e^{x_1} + \e^{x_2} - K)^+,
\end{align}
whose Fourier transform is
\begin{align}
\widehat{g}(z) = K^{1+iz_1+iz_2}
  \frac{\Gamma(iz_2)\Gamma(1-iz_1-iz_2)}{\Gamma(1-iz_1)},
\end{align}
for $z\in\mathcal{Y}:=\{z\in\C^2:\Im z_2 <0, \Im(z_1+z_2)>1\}$; see
\cite{HubalekKallsen03} and \cite{HurdZhou09}. Here, $\Gamma$ denotes the
Gamma function. Therefore, using also
\cite[Thm.~3.2]{EberleinGlauPapapantoleon08}, we have that the price of an
option on a defaultable zero coupon bond admits the following semi-analytical
expression:
\begin{align}
\pi_0^{\text{CO}}
 &= \frac{\overline{B}(0,T_i)}{4\pi^2} \int_{\R^2} \widehat{g}(iR-w)
     M_Y(R+iw) \dd w,
\end{align}
for $iR\in\mathcal{Y}$ such that $M_Y(R)<\infty$, where $M_Y$ denotes the moment
generating function of the random vector $Y=(Y_1,Y_2)$. This can be computed
explicitly using the affine property of the driving process $X$ under the
measure $\overline\P_i$. We have
\begin{align}
M_Y(w)
 &= \overline\E_i \big[\e^{\scal{w}{Y}}\big]
  = \overline\E_i \big[\e^{w_1Y_1 + w_2Y_2}\big] \nonumber\\
 &= \mathcal{X}\, \overline\E_i
    \big[\e^{\scal{w_1 B^m_i + w_2 \overline{B}^m_i}{X_{T_i}}}\big]
    \nonumber\\
 &= \mathcal{X} \exp\left\{
     \overline\phi^i_{T_i}(w_1 B^m_i + w_2 \overline{B}^m_i)
   + \big\langle\overline\psi^i_{T_i}(w_1 B^m_i + w_2 \overline{B}^m_i),
     X_0\big\rangle \right\},
\end{align}
where
\begin{align}
\mathcal{X}
 &= \exp\big\{ w_1(\log\pi + A^m_i)
     + w_2(\log(1-\pi) + \overline{A}^m_i) \big\}.
\end{align}

\begin{remark}
Similar semi-analytical expressions can be derived for other derivatives with
non-linear payoffs, such as credit spread options. Credit default swaptions,
also known as CDS options, are more difficult to handle, but expressions
involving a high-dimensional integration can still be derived; see also
\cite[\S7.3]{KellerResselPapapantoleonTeichmann09}.
\end{remark}

\subsection{Vulnerable options}

The term \textit{credit risk} applies to two different types of risk:
\textit{reference risk} and \textit{counterparty risk}. Reference risk is the
risk associated with an underlying asset (reference) in a contract, whereas
counterparty risk refers to any kind of risk associated with either of the
counterparties involved in a contract; see Figure \ref{cc} for a graphical
representation. Contingent claims with reference risk traded over-the-counter
between default-free parties are labeled \emph{credit derivatives} and the
collective name \emph{vulnerable claims} refers to contingent claims traded
over-the-counter between default-prone parties with an underlying asset that is
assumed to be default-free. The name \emph{vulnerable} goes back to
\cite{JohnsonStulz87},  who studied the impact of default risk of an option
writer on option prices. We mention also some other papers studying counterparty
risk such as  \cite{JarrowTurnbull95}, \cite{HullWhite95}, \cite{HugeLando99}
and \cite{LotzSchloegl00}.

In this section, we study an application of the defaultable affine LIBOR model
to the pricing of vulnerable options. Again, we obtain an analytical expression
for the price of a vulnerable option which does not involve any approximations;
compare with \cite[Section 4.3]{Grbac10} where vulnerable options are studied
in the defaultable L\'evy LIBOR framework which requires ``frozen drift"-type
approximations.

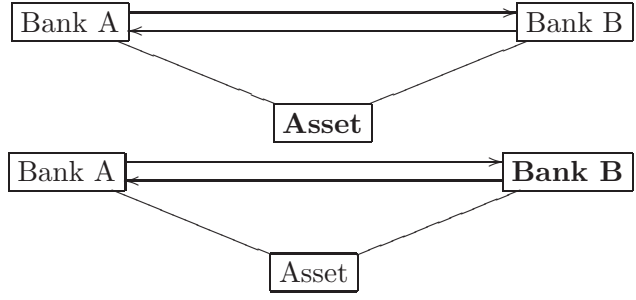
\begin{figure}
\begin{align*}
\xymatrix{
*+[F]\txt{Bank A}\ar@<.75ex>[rrrr] \ar@{-}[drr]
&&&& *+[F]\txt{Bank B} \ar@<.75ex>[llll] \ar@{-}[dll]\\
&& *+[F]{\txt{\textbf{Asset}}} && \\}\\
\xymatrix{
*+[F]\txt{Bank A}\ar@<.75ex>[rrrr] \ar@{-}[drr]
&&&& *+[F]{\txt{\textbf{Bank B}}} \ar@<.75ex>[llll] \ar@{-}[dll]\\
&& *+[F]\txt{Asset} && \\}
\end{align*}
\caption{A graphical illustration of reference risk (top) vs. counterparty risk
(bottom). The bold letters indicate where default risk lies in a contract
between counterparties A and B.}
\label{cc}
\end{figure}

A vulnerable European call option with maturity $T_k$ and strike $K$ on a
default-free bond $B(\cdot, T_m)$, where $T_m \geq T_k$, has a payoff  given by
\begin{align*}
\overline{C}_{T_{k}}
 &= C_{T_{k}} \ifnodefault{T_k} + q C_{T_{k}} \ifdefault{T_k}\\
 &= (C_{T_{k}}-q C_{T_{k}}) \ifnodefault{T_k} + q C_{T_{k}},
\end{align*}
where $C_{T_{k}}=(B(T_{k},T_{m})-K)^{+} $ is the payoff at maturity $T_k$ of a
European call option written on a default-free bond with maturity $T_m$, $\tau$
is the default time of the writer of the option and $q$ the recovery rate (in
case of default the payoff of the option at maturity is reduced by a factor $q
\in[0,1]$).

Therefore, we give here a new interpretation to defaultable bonds
$B^{0}(\cdot,T_k)$, which are now assumed to be issued by the \emph{writer of
the vulnerable option}. Using the definition of the forward \lib rate, we can
rewrite the payoff of the vulnerable option as follows
\begin{eqnarray*}
\overline{C}_{T_{k}} & = & \ifnodefault{T_k} (1-q) (B(T_{k}, T_{m})-K)^{+} + q
(B(T_{k}, T_{m})-K)^{+} \\
& = & \ifnodefault{T_k} (1-q) \left(\prod_{l=k}^{m-1}(1+\delta_{l}L(T_{k},
T_{l}))^{-1}-K\right)^{+} \\
&& \quad  + \, q \left(\prod_{l=k}^{m-1}(1+\delta_{l}L(T_{k},
T_{l}))^{-1}-K\right)^{+}.
\end{eqnarray*}
Its value at time $t=0$ is given by
\begin{align*}
\overline{C}_{0}
&= B(0, T_{k})\, \E_{k} \left[\ifnodefault{T_k} (1-q)
 \left(\prod_{l=k}^{m-1}(1+\delta_{l}L(T_{k}, T_{l}))^{-1}-K\right)^{+} \right.
 \\
& \qquad\qquad\qquad +  \left. q \left(\prod_{l=k}^{m-1}(1+\delta_{l}L(T_{k},
   T_{l}))^{-1}-K\right)^{+}\right] \\
&= \overline{B}(0, T_{k})\, \overline{\E}_{k}  \left[(1-q)
 \left(\prod_{l=k}^{m-1}(1+\delta_{l}L(T_{k}, T_{l}))^{-1}-K\right)^{+} \right]
 \\
& \qquad +  B(0, T_{k})\,
    \E_{k}\left[q \left(\prod_{l=k}^{m-1}(1+\delta_{l}L(T_{k},
      T_{l}))^{-1}-K\right)^{+} \right]\\
& =  \overline{B}(0, T_{k})\, (1-q)\, \overline{\E}_{k} \left[ \left( \e^Z- K
 \right)^{+}\right]  + B(0, T_k)\, q\, \E_{k} \left[ \left( \e^Z- K
 \right)^{+}\right],
\end{align*}
where
\begin{align*}
Z & := A^m_k + \scal{B^m_k}{X_{T_k}},
\end{align*}
see equation \eqref{product-Libor}. The payoff function of a call option
\begin{align*}
g(y) & = (\e^y- K)^+
\end{align*}
has the Fourier transform
\begin{align*}
\widehat{g}(z)& = \frac{K^{1+iz}}{iz(1+iz)},
\end{align*}
for $z\in\mathcal{Z}:=\{z\in\C: \Im z>1\}$. Therefore, using
\cite[Thm.~2.2]{EberleinGlauPapapantoleon08}, we obtain
\begin{align*}
\overline{C}_0 & = \frac{\overline{B}(0, T_{k}) (1-q)}{2 \pi} \int_{\R}
\widehat{g}(i R_1 - v) M_Z^{\overline{\P}_k} (R_1 + iv) \dv  \\
& \ \ +  \frac{B(0, T_{k}) q}{2 \pi} \int_{\R} \widehat{g}(i R_2 - w) M_Z^{\P_k}
(R_2 + iw) \ud w,
\end{align*}
for $iR_1,iR_2\in\mathcal{Z}$ such that $M_Z^{\overline{\P}_k}(R_1)<\infty$ and
$M_Z^{\P_k}(R_2)<\infty$. The moment generating function of $Z$ under
$\overline{\P}_k$ and $\P_k$ respectively, is provided by
\begin{align*}
M_Z^{\overline{\P}_k} (v)
& = \exp \left\{ v A^m_{k} +  \overline\phi^k_{T_k}(v B^m_{k})
   + \big\langle\overline\psi^k_{T_k}(v B^m_{k}),
     X_0\big\rangle \right\}
\end{align*}
and
\begin{align*}
M_Z^{\P_k} (w)
& = \exp \left\{ w A^m_{k} +  \phi^k_{T_k}(w B^m_{k})
   + \big\langle\psi^k_{T_k}(w B^m_{k}),
     X_0\big\rangle \right\}.
\end{align*}

\bibliographystyle{alpha}
\bibliography{references}

\begin{thebibliography}{KRPT11}

\bibitem[BGM97]{BraceGatarekMusiela97}
A.~Brace, D.~G\c{a}tarek, and M.~Musiela.
\newblock {The market model of interest rate dynamics}.
\newblock {\em Math. Finance}, 7:127--155, 1997.

\bibitem[BM06]{BrigoMercurio06}
D.~Brigo and F.~Mercurio.
\newblock {\em {Interest Rate Models: Theory and Practice}}.
\newblock Springer, 2nd edition, 2006.

\bibitem[BmY78]{BremaudYor78}
P.~Br{\'e}\-maud and M.~Yor.
\newblock {Changes of filtrations and of probability measures}.
\newblock {\em Z. Wahrsch. Verw. Geb.}, 45:269--295, 1978.

\bibitem[BR02]{BieleckiRutkowski02}
T.~R. Bielecki and M.~Rutkowski.
\newblock {\em {Credit Risk: Modeling, Valuation and Hedging}}.
\newblock Springer, 2002.

\bibitem[CGN12]{Crepey_Grbac_Nguyen_2011}
S.~Cr{\'e}pey, Z.~Grbac, and H.-N. Nguyen.
\newblock {A multiple-curve HJM model of interbank risk}.
\newblock {\em Math. Financ. Econ.}, 6:155--190, 2012.

\bibitem[DFS03]{DuffieFilipovicSchachermayer03}
D.~Duffie, D.~Filipovi{\'c}, and W.~Schachermayer.
\newblock {Affine processes and applications in finance}.
\newblock {\em Ann. Appl. Probab.}, 13:984--1053, 2003.

\bibitem[DG01]{Duffie_Garleanu_2001}
D.~Duffie and N.~G{\^a}rleanu.
\newblock {Risk and valuation of collateralized debt obligations}.
\newblock {\em Financial Analysts J.}, 57:41--59, 2001.

\bibitem[EGP10]{EberleinGlauPapapantoleon08}
E.~Eberlein, K.~Glau, and A.~Papapantoleon.
\newblock {Analysis of {F}ourier transform valuation formulas and
  applications}.
\newblock {\em Appl. Math. Finance}, 17:211--240, 2010.

\bibitem[EJY00]{ElliottJeanblancYor00}
R.~Elliott, M.~Jeanblanc, and M.~Yor.
\newblock {On models of default risk}.
\newblock {\em Math. Finance}, 10:179--195, 2000.

\bibitem[EKS06]{EberleinKlugeSchoenbucher06}
E.~Eberlein, W.~Kluge, and Ph.~J. Sch{\"o}nbucher.
\newblock {The {L}{\'e}vy {LIBOR} model with default risk}.
\newblock {\em J. Credit Risk}, 2:3--42, 2006.

\bibitem[Fil01]{Filipovic_2001}
D.~Filipovi{\'c}.
\newblock {A general characterization of one factor affine term structure
  models}.
\newblock {\em Finance Stoch.}, 5:389--412, 2001.

\bibitem[Fil05]{Filipovic05}
D.~Filipovi{\'c}.
\newblock {Time-inhomogeneous affine processes}.
\newblock {\em Stochastic Process. Appl.}, 115:639--659, 2005.

\bibitem[FT11]{Filipovic_Trolle_2011}
D.~Filipovi{\'c} and A.~Trolle.
\newblock {The term structure of interbank risk}.
\newblock Preprint, EPFL, 2011.

\bibitem[GBM06]{GatarekBachertMaksymiuk06}
D.~Gatarek, P.~Bachert, and R.~Maksymiuk.
\newblock {\em {The LIBOR Market Model in Practice}}.
\newblock Wiley, 2006.

\bibitem[Grb10]{Grbac10}
Z.~Grbac.
\newblock {\em {Credit risk in {L}{\'e}vy {LIBOR} modeling: rating based
  approach}}.
\newblock PhD thesis, Univ. Freiburg, 2010.

\bibitem[HK05]{HubalekKallsen03}
F.~Hubalek and J.~Kallsen.
\newblock {Variance-optimal hedging and {M}arkowitz-efficient portfolios for
  multivariate processes with stationary independent increments with and
  without constraints}.
\newblock Working paper, TU M{\"u}nchen, 2005.

\bibitem[HL99]{HugeLando99}
B.~Huge and D.~Lando.
\newblock {Swap pricing with two-sided default risk in a rating-based model}.
\newblock {\em European Finance Rev.}, 3:239--268, 1999.

\bibitem[HW95]{HullWhite95}
J.~Hull and A.~White.
\newblock {The impact of default risk on the prices of options and other
  derivative securities}.
\newblock {\em J. Banking Finance}, 19:299--322, 1995.

\bibitem[HZ10]{HurdZhou09}
T.~R. Hurd and Z.~Zhou.
\newblock {A {F}ourier transform method for spread option pricing}.
\newblock {\em SIAM J. Financial Math.}, 1:142--157, 2010.

\bibitem[Jam97]{Jamshidian97}
F.~Jamshidian.
\newblock {LIBOR and swap market models and measures}.
\newblock {\em Finance Stoch.}, 1:293--330, 1997.

\bibitem[JL08]{JeanblancLeCam08a}
M.~Jeanblanc and Y.~{Le Cam}.
\newblock {Reduced form modelling for credit risk}.
\newblock Working paper, 2008.

\bibitem[JR00]{JeanblancRutkowski00}
M.~Jeanblanc and M.~Rutkowski.
\newblock {Modelling of default risk: an overview}.
\newblock In {\em {Mathematical Finance: Theory and Practice}}. Higher
  education press, Beijing, 2000.

\bibitem[JS87]{JohnsonStulz87}
H.~Johnson and R.~Stulz.
\newblock {The pricing of options with default risk}.
\newblock {\em J. Finance}, 42:267--280, 1987.

\bibitem[JS03]{JacodShiryaev03}
J.~Jacod and A.~N. Shiryaev.
\newblock {\em {Limit Theorems for Stochastic Processes}}.
\newblock Springer, 2nd edition, 2003.

\bibitem[JT95]{JarrowTurnbull95}
R.~A. Jarrow and S.~M. Turnbull.
\newblock {Pricing derivatives on financial securities subject to credit risk}.
\newblock {\em J. Finance}, 50:53--85, 1995.

\bibitem[Klu05]{Kluge05}
W.~Kluge.
\newblock {\em {Time-inhomogeneous L{\'e}vy processes in interest rate and
  credit risk models}}.
\newblock PhD thesis, Univ. Freiburg, 2005.

\bibitem[KN10]{KokholmNicolato10}
T.~Kokholm and E.~Nicolato.
\newblock {Sato processes in default modelling}.
\newblock {\em Appl. Math. Finance}, 17:377--397, 2010.

\bibitem[KR08]{KellerRessel08}
M.~Keller-Ressel.
\newblock {\em {Affine processes: Theory and applications to finance}}.
\newblock PhD thesis, TU Vienna, 2008.

\bibitem[KRPT11]{KellerResselPapapantoleonTeichmann09}
M.~Keller-Ressel, A.~Papapantoleon, and J.~Teichmann.
\newblock {The affine {LIBOR} models}.
\newblock {\em Math. Finance}, 2011.
\newblock (forthcoming).

\bibitem[LS00]{LotzSchloegl00}
C.~Lotz and L.~Schl{\"o}gl.
\newblock {Default risk in a market model}.
\newblock {\em J. Banking Finance}, 24:301--327, 2000.

\bibitem[Pap10]{Papapantoleon10b}
A.~Papapantoleon.
\newblock {Old and new approaches to {LIBOR} modeling}.
\newblock {\em Stat. Neerlandica}, 64:257--275, 2010.

\bibitem[PSS12]{PapapantoleonSchoenmakersSkovmand10}
A.~Papapantoleon, J.~Schoenmakers, and D.~Skovmand.
\newblock {Efficient and accurate log-{L}{\'e}vy approximations to
  {L}{\'e}vy-driven {LIBOR} models}.
\newblock {\em J. Comput. Finance}, 15(4):3--44, 2012.

\bibitem[Sch00]{Schoenbucher00}
P.~J. Sch{\"o}nbucher.
\newblock {A LIBOR market model with default risk}.
\newblock Working paper, University of Bonn, 2000.

\bibitem[SSM95]{SandmannSondermannMiltersen95}
K.~Sandmann, D.~Sondermann, and K.~R. Miltersen.
\newblock {Closed form term structure derivatives in a {Heath--Jarrow--Morton}
  model with log-normal annually compounded interest rates}.
\newblock In {\em {Proceedings of the Seventh Annual European Futures Research
  Symposium Bonn}}, pages 145--165, 1995.
\newblock Chicago Board of Trade.

\end{thebibliography}

\end{document}